\documentclass[pdflatex,sn-mathphys-ay,iicol]{sn-jnl}

\usepackage{graphicx}%
\usepackage{multirow}%
\usepackage{amsmath,amssymb,amsfonts}%
\usepackage{amsthm}%
\usepackage{mathrsfs}%
\usepackage[title]{appendix}%
\usepackage{xcolor}%
\usepackage{textcomp}%
\usepackage{manyfoot}%
\usepackage{booktabs}%
\usepackage{algorithm}%
\usepackage{algorithmicx}%
\usepackage{algpseudocode}%
\usepackage{listings}%

\usepackage{calrsfs}
\usepackage{subcaption}
\RequirePackage{tikz}

\newcommand{\reals}{\mathbb{R}}
\newcommand{\integers}{\mathbb{Z}}

\newcommand{\X}{{\mathcal{X}}}

\newcommand{\T}{\mathcal{T}}
\renewcommand{\O}{\mathcal{O}}

\newcommand{\Prob}{\mathbb{P}}

\theoremstyle{thmstyleone}%
\newtheorem{theorem}{Theorem}
\newtheorem{proposition}[theorem]{Proposition}%
\newtheorem{lemma}[theorem]{Lemma}
\newtheorem{corollary}[theorem]{Corollary}

\theoremstyle{thmstylethree}%
\newtheorem{example}{Example}%

\theoremstyle{thmstylethree}%
\newtheorem{definition}{Definition}%

\begin{document}

\title[Article Title]{Random sampling of contingency tables and partitions: Two practical examples of the Burnside process}


\author[1,2]{\fnm{Persi} \sur{Diaconis}}
\author*[2]{\fnm{Michael} \sur{Howes}}\email{mhowes@stanford.edu}

\affil[1]{\orgdiv{Department of Mathematics}, \orgname{Stanford University}, \orgaddress{\state{California}, \country{USA}}}

\affil[2]{\orgdiv{Department of Statistics}, \orgname{Stanford University}, \orgaddress{\state{California}, \country{USA}}}


\abstract{
    This paper gives new, efficient algorithms for approximate uniform sampling of contingency tables and integer partitions. The algorithms use the Burnside process, a general algorithm for sampling a uniform orbit of a finite group acting on a finite set. We show that a technique called `lumping' can be used to derive efficient implementations of the Burnside process. For both contingency tables and partitions, the lumped processes have far lower per step complexity than the original Markov chains. We also define a second Markov chain for partitions called the reflected Burnside process. The reflected Burnside process maintains the computational advantages of the lumped process but empirically converges to the uniform distribution much more rapidly. By using the reflected Burnside process we can easily sample uniform partitions of size $10^{10}$.}

\keywords{Markov chain Monte Carlo, contingency tables, partitions, Burnside process}


\maketitle

\section{Introduction}\label{sec:intro}

The Burnside process is a Markov chain for sampling combinatorial objects that arise in applied probability and statistics. We begin with a general description and then specialize to examples. Let $\X$ be a finite set and $G$ a group acting on $\X$. We will write $x^g$ for the image of $g$ acting on $x$. Under the group action, the set $\X$ splits into disjoint orbits
\[
\X =\O_{x_1} \sqcup \O_{x_2} \sqcup \cdots \sqcup \O_{x_Z},
\]
where $\O_x=\{x^g : g \in G\}$ is the orbit containing $x$ and $Z$ is the number of orbits. The orbit $\O_x$ can be thought of as the element $x \in \X$ `up to symmetry.' Natural questions about the orbits are
\begin{itemize}
    \item How many orbits are there?
    \item What is the size of a typical orbit?
    \item How can an orbit be uniformly chosen?
\end{itemize}
Solving the third problem helps with the first two. The Burnside process \citep{jerrum1993uniform} is a Markov chain on $\X$ with stationary distribution 
\[
    \pi(x) = \frac{1}{Z|\O_x|}.
\] 
This means that if $X$ is sampled from $\pi$, then the orbit containing $X$ is uniform over the $Z$ possible orbits. Furthermore, starting at any $x \in \X$ and taking sufficiently many steps of the Burnside process will produce a random $X \in \X$ that is approximately distributed according to $\pi$. An important feature of the Burnside process is that it can be `lumped to orbits' \citep{diaconis2005analysis}. This means that if $(X_i)_{i=0}^n$ is a Markov chain on $\X$ that evolves according to the Burnside process, then the orbit process $(\O_{X_i})_{i=0}^n$ is also a Markov chain. In other words, the Burnside process can be run directly on the space of orbits. Special cases of the Burnside process can be used to sample P\'olya trees \citep{bartholdi2024algorithm}, set partitions \citep{paguyo2022mixing}, conjugacy classes in groups \citep{rahmani2022mixing, diaconis2025counting}, Bose--Einstein statistics \citep{diaconis2005analysis, diaconis2020hahn} and graph colorings for lifted probabilistic inference \citep{holtzen2020generating}. These examples, as well as general facts about the Burnside process and group actions are reviewed in Section~\ref{sec:Burnside}. 

In Sections~\ref{sec:tables} and \ref{sec:parts} we explain how the Burnside process can be used to approximately sample contingency tables and partitions--leading to applications in combinatorics and statistics. We describe both the original processes and their lumped versions. Our main finding is that the lumped Burnside process can be enormously more efficient than the original process. For contingency tables, the lumped process is exponentially faster than the original process. For partitions, the speed-up is quadratic. 

We also introduce a variation on the Burnside process called the \emph{reflected Burnside process}. This is a second Markov chain for partitions that combines the lumped Burnside process with a `deterministic jump.' The reflected Burnside process appears to mix incredibly rapidly, and empirically generates uniform partitions of size $10^{10}$ in around 50 steps. The advantages of the reflected Burnside process over the original Burnside process are discussed in Section~\ref{sec:transposing} and shown in Figure~\ref{fig:alternating-parts}.

The final section summarizes our experience with implementing the Burnside process. We explain why the lumped process is so much more efficient than the original process and give pointers to other applications of the Burnside process where lumping could help. Our implementation of the Burnside process for partitions and contingency tables is available in Julia \citep{bezanson2017julia} at \url{https://github.com/Michael-Howes/BurnsideProcess/}. The repository also contains notebooks for recreating the figures in this article.

\section{Background on group actions and the Burnside process}\label{sec:Burnside}

Many statisticians know group actions in the context of exchangeable random variables \citep{aldous1985exchangeability} or equivariant estimators \citep[Chapter~3]{lehmann1998theory}. In this section, we will give a formal introduction to group actions and their use in the Burnside process. Let $\X$ be a finite set and $G$ a finite group. A \emph{group action} of $G$ on $\X$ is a map from $\X \times G$ to $\X$ written $(x,g) \mapsto x^g$ satisfying the following two properties:
\begin{enumerate}
    \item For all $x \in \X$, $x^\mathrm{Id} = x$ where $\mathrm{Id}$ is the identity element of $G$.
    \item For all $g,h \in G$ and $x \in \X$, $(x^g)^h = x^{gh}$.
\end{enumerate}
The group action of $G$ on $\X$ defines an equivalence relation on $\X$ with $x \sim y$ if and only if $y = x^g$ for some $g \in G$. The equivalence classes are called \emph{orbits} and the orbit containing $x$ is denoted by $\O_x = \{x^g : g \in G\}$. The set $\X / G = \{ \O_x : x \in \X\}$ represents the space of all orbits.

Associated to every $x \in \X$ there is subgroup of $G$ containing all group element that fix $x$. This subgroup is called the stabilizer of $x$ and is written as $G_x = \{g \in G : x^g = x\}$. The \emph{orbit--stabilizer theorem} states that for all $x \in \X$,
\begin{equation}
    |G| = |\O_x||G_x|. \label{eq:orbit-stab}
\end{equation}
Finally, for each $g \in G$, we will let $\X^g = \{x \in \X : x^g = x\}$ denote the set of points fixed by $g$. For a clear, textbook account of the needed group theory, see \citet[Chapter~1, Section~7]{suzuki1986group}

The Burnside process was introduced by \citet{jerrum1993uniform} and developed for `computational P{\'o}lya theory' in \citet{goldberg1993automating, goldberg2002burnside}. The Burnside process is a Markov chain on $\X$. A single step of this Markov chain transitions from $x$ to $y$ as follows:
\begin{enumerate}
    \item From $x \in \X$, uniformly choose $g \in G_x = \{g \in G : x^g = x\}$.
    \item From $g \in G$, uniformly choose $y \in \X^g= \{y\in \X : y^g = y\}$. 
\end{enumerate}
The probability of moving from $x$ to $y$ is therefore 
\begin{eqnarray}
    P(x,y) = \frac{1}{|G_x|}\sum_{g \in G_x \cap G_y} \frac{1}{|\X^g|}.\label{eq:kernel}
\end{eqnarray}
As claimed in the introduction and shown in \citet{jerrum1993uniform}, this $P$ is an ergodic, reversible Markov kernel with stationary distribution $\pi(x) = \frac{1}{Z|\O_x|}$. To see this, note that by \eqref{eq:orbit-stab},
\begin{align*}
 \pi(x)P(x,y) &= \frac{1}{Z|\O_x| |G_x|}\sum_{g \in G_x \cap G_y} \frac{1}{|\X^g|}\\
 & = \frac{1}{Z|G|}\sum_{g \in G_x \cap G_y} \frac{1}{|\X^g|}\\
 & = \pi(y)P(y,x).
\end{align*}
This shows that $P$ is reversible with respect to $\pi$. Furthermore, we have that $P(x,y) > 0$ for all $x,y \in \X$. This is because it is always possible to transition from $x$ to $y$ by choosing the identity group element. This implies that $P$ is ergodic.  It follows that the Burnside process can be used to approximately sample uniform orbits, provided that the Markov chain has been run for sufficiently many steps.

The following examples give a feel for the two steps in the Burnside process.
\begin{example}[Bose--Einstein statistics]\label{ex:C2}
    Let $\X=C_2^n$ be the set of binary $n$-tuples. Let $G=S_n$ act on $\X$ by permuting coordinates. There are $n+1$ orbits with
    \[ 
        \O_i = \{x : |x|=i\} \text{ for } 0 \le i \le n,
    \]
    where $|x|$ is the number of $1$'s in $x$.  The Burnside process proceeds as follows:
    \begin{enumerate}
        \item From $x$, choose $\sigma \in S_n$ fixing $x$. Such a permutation must permute the $1$'s and $0$'s in $x$ amongst themselves. Thus, if $|x|=i$, then $G_x \cong S_i \times S_{n-i}$. It is easy to sample from $S_i \times S_{n-i}$ and hence $G_x$.
        \item From $\sigma$, we must choose $y \in C_2^n$ fixed by $\sigma$. This is also easy to do. First, decompose $\sigma$ into disjoint cycles $\sigma=c_1c_2\cdots c_j$. Then, label each cycle $0$  or $1$ with probability $1/2$ independently. The label of $y_i$ is determined by the cycle containing $i$. 
    \end{enumerate}
\end{example}

\begin{example}[Random matrices]
    Let $U_n(q)$ be the group of uni-upper triangular matrices with entries in the finite field $\mathbb{F}_q$. When $n$ is at all large (e.g. $n=20$) the conjugacy classes of $U_n(q)$ are unknown and, in a sense, unknowable. In representation theory, classification problems can either be `tame' or `wild', and the wild problems are hopelessly hard \citep{donovan1972some, drozd1980tame}. \citet{gudivok1990classes} showed that classifying the conjugacy classes of $U_n(q)$ is wild. Section~5 of \citet{diaconis2021complexity} gives more background and connects the complexity of $U_n(q)$ to theory of random graphs. 
    
    The Burnside process gives a way of sampling a uniform conjugacy class. The equivalence between sampling and counting \citep{jerrum1986random,broder1986hard} gives a route to effective enumeration of $|C_n(q)|$--the number of conjugacy classes. This program has be carried out in \citep{diaconis2025counting} who get useful estimates of $|C_n(2)|$ and $|C_n(3)|$ for $n$ up to $40$.

    Here $\X=G=U_n(q)$ and $X^M = M^{-1}XM$ so that $U_n(q)$ acts on itself by conjugation. The Burnside process runs on $\X$. For the first step, from a matrix $X \in U_n(q)$ one must sample $M \in U_n(q)$ so that $MX=XM$. For fixed $M$ this is a linear algebra problem over $\mathbb{F}_q$. Gaussian elimination can be used to find a basis for the solution space. Choosing a random linear combination of basis vectors gives a uniform choice of $M$. The second step is the same as the first (choose $X'$ with $X'M= MX'$). Of course, this is a serious undertaking, but it is quite possible and has been usefully carried out \citep{diaconis2025counting}. It provides an important supplement (and check) on exact computation.
\end{example}
\begin{example}[P\'olya trees]\label{ex:trees}
    Let $\mathcal{C}_n$ be the set of labeled trees on $n$ vertices rooted at the vertex 1. Cayley's formula shows $|\mathcal{C}_n|=n^{n-2}$ (so $|\mathcal{C}_4|=16$, see Figure~\ref{fig:cayley-trees}). The permutation group $S_{n-1}$ acts on $\mathcal{C}_n$ by permuting the vertices and fixing the roots. The orbits of $S_{n-1}$ on $\mathcal{C}_n$ are unlabeled P\'olya trees $\mathcal{T}_n$ (see $\mathcal{T}_4$ in Figure~\ref{fig:unlabeled-trees}). There is no formula for $|\mathcal{T}_n|$ and enumerative questions are challenging \citep{drmota2009random}. 
    
    Consider the Burnside process with $\X =\mathcal{C}_n$ and $G=S_{n-1}$. The two steps are:
    \begin{enumerate}
        \item From $t \in \mathcal{C}_n$, choose $\sigma \in S_{n-1}$ uniformly with $t^\sigma = t$.
        \item From $\sigma$, choose $t_1$ uniformly with $t_1^\sigma = t_1$.
    \end{enumerate}
    Both steps are challenging, in theory and practice. \citet{bartholdi2024algorithm} successfully implement the Burnside process. This requires an extension of Cayley's formula to a formula for the number of labeled trees fixed by a given permutation, and uses substantial input from the computational graph theory community \citep{anders2021parallel,anders2023engineering}  (and their accompanying software \href{https://automorphisms.org/}{https://automorphisms.org/}). The implementation efficiently yields samples from trees of size 10,000. The simulations show that published theorems were `off' and suggested that the limit distributions needed location and scaling shifts.
\end{example}

\def\treeA#1#2#3#4{\node[fill, label=right:{#1}] (a) at (0,0) {};
  \node[label=right:{#2}] (b) at (0,-0.5) {};
  \node[label=right:{#3}] (c) at (0,-1) {};
  \node[label=right:{#4}] (d) at (0,-1.5) {};
  \draw (a) -- (b) -- (c) -- (d);}
\def\treeB#1#2#3#4{\node[fill, label=right:{#1}] (a) at (0,-0.1) {};
  \node[label=right:{#2}] (b) at (0,-0.7) {};
  \node[label=right:{#3}] (c) at (-0.3,-1.3) {};
  \node[label=right:{#4}] (d) at (0.3,-1.3) {};
  \draw (a) -- (b) -- (c) (b) -- (d);}
\def\treeC#1#2#3#4{\node[fill, label=right:{#1}] (a) at (0,-0.1) {};
  \node[label=right:{#2}] (b) at (-0.3,-0.7) {};
  \node[label=right:{#3}] (c) at (0.3,-0.7) {};
  \node[label=right:{#4}] (d) at (0.3,-1.3) {};
  \draw (a) -- (b) (a) -- (c) -- (d);}
\def\treeD#1#2#3#4{\node[fill, label=right:{#1}] (a) at (0,-0.3) {};
  \node[label=right:{#2}] (b) at (-0.5,-1.1) {};
  \node[label=right:{#3}] (c) at (0,-1.1) {};
  \node[label=right:{#4}] (d) at (0.5,-1.1) {};
  \draw (a) -- (b) (a) -- (c) (a) -- (d);}

\begin{figure*}
  \centerline{\begin{tikzpicture}[every node/.style={circle,inner sep=0pt,minimum size=4pt,draw}]
    \begin{scope}[xshift=0cm]\treeA1234\end{scope}
    \begin{scope}[xshift=1cm]\treeA1243\end{scope}
    \begin{scope}[xshift=2cm]\treeA1324\end{scope}
    \begin{scope}[xshift=3cm]\treeA1342\end{scope}
    \begin{scope}[xshift=4cm]\treeA1423\end{scope}
    \begin{scope}[xshift=5cm]\treeA1432\end{scope}
    \begin{scope}[xshift=6.2cm]\treeB1234\end{scope}
    \begin{scope}[xshift=7.6cm]\treeB1324\end{scope}
    \begin{scope}[xshift=9cm]\treeB1423\end{scope}
    \begin{scope}[yshift=-2cm,xshift=1cm]\treeC1234\end{scope}
    \begin{scope}[yshift=-2cm,xshift=2.2cm]\treeC1243\end{scope}
    \begin{scope}[yshift=-2cm,xshift=3.4cm]\treeC1324\end{scope}
    \begin{scope}[yshift=-2cm,xshift=4.6cm]\treeC1342\end{scope}
    \begin{scope}[yshift=-2cm,xshift=5.8cm]\treeC1423\end{scope}
    \begin{scope}[yshift=-2cm,xshift=7cm]\treeC1432\end{scope}
    \begin{scope}[yshift=-2cm,xshift=8.6cm]\treeD1234\end{scope}
  \end{tikzpicture}}
\caption{The $16$ labeled rooted trees at $1$, for $n=4$. Here and below the root vertex is indicated as solid, and non-root vertices are hollow. Reproduced from \citet{bartholdi2024algorithm}.}
\label{fig:cayley-trees}
\end{figure*}
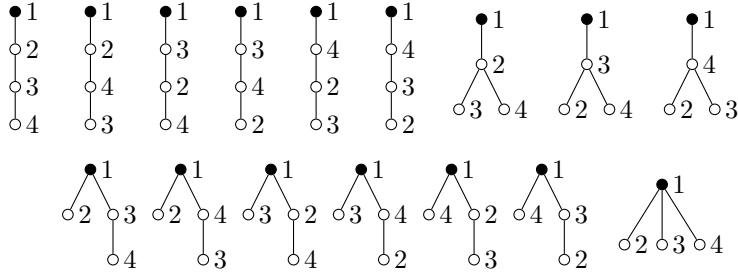

\begin{figure*}
  \centerline{\begin{tikzpicture}[every node/.style={circle,inner sep=0pt,minimum size=4pt,draw}]
    \begin{scope}[xshift=0cm]\treeA{}{}{}{}\end{scope}
    \begin{scope}[xshift=2cm]\treeB{}{}{}{}\end{scope}
    \begin{scope}[xshift=4cm]\treeC{}{}{}{}\end{scope}
    \begin{scope}[xshift=6cm]\treeD{}{}{}{}\end{scope}
  \end{tikzpicture}}
\caption{The $4$ unlabeled rooted trees, for $n=4$. Reproduced from \citet{bartholdi2024algorithm}.}\label{fig:unlabeled-trees}
\end{figure*}
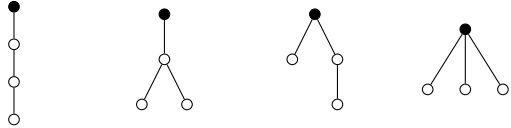

\subsection{Lumping the Burnside process}

\citet{diaconis2005analysis} showed that the Burnside process can be `lumped to orbits.' More precisely, let $\overline{P}$ be a Markov kernel on $\X/G=\{\O_x : x \in \X\}$ given by
\begin{equation}
    \overline{P}(\O_x,\O_y) = \sum_{z \in \O_y}P(x,z),\label{eq:lumped-kernel}
\end{equation}
where $P$ is the kernel for the Burnside process in \eqref{eq:kernel}. \citet{diaconis2005analysis} showed the following.
\begin{proposition}\label{prop:lump}
    Let $P$ and $\overline{P}$ be as above. Then
    \begin{enumerate}
        \item  The definition of $\overline{P}$ does not depend on the choice of $x \in \O_x$. 
        \item The kernel $\overline{P}$ is ergodic and symmetric. Thus, the uniform distribution on $\X/G$ is the unique invariant distribution of $\overline{P}$. 
        \item  If $(X_i)_{i=0}^n$ is a Markov chain on $\X$ with transition kernel $P$, then $(\O_{X_i})_{i=0}^n$ is a Markov chain on $\X/G$ with transition kernel $\overline{P}$.
    \end{enumerate}
\end{proposition}
For completeness, a proof of Proposition~\ref{prop:lump} is given in Appendix~\ref{appn:lump}. It can take work to translate the definition of $\overline{P}$ in \eqref{eq:lumped-kernel} into a useful formula or description for the chance of transitioning from one orbit to another. 
\begin{example}[Example~\ref{ex:C2} continued]\label{ex:C2-revisited}
    For $\X=C_2^n$ and $G=S_n$, the orbits are indexed by $\{0,1,\ldots,n\}$. Let $\overline{P}(j,k)$ be the chance of moving from $j$ to $k$ in one step of the lumped Burnside process. The following result is given in Section~3 of \citet{diaconis2005analysis}. There it is used to bound the mixing time of the Burnside process for $S_n$ acting on $C_2^n$.  

    Let 
        \[
            \alpha^n_k = \binom{2k}{k}\binom{2(n-k)}{n-k}\slash  2^{2n}  \quad 0 \le k  \le n,
        \]
        be the discrete arcsine distribution \citep[Chapter~3]{feller1968introduction}. Then
    \begin{align}\label{eq:C2-lumped}
        \overline{P}(0,k) & = \overline{P}(n,k) = \alpha^n_k,\nonumber\\
        \overline{P}(j,k) &=\sum_{l=(j-k-m)_+}^{\min\{j,k\}}\alpha^j_l \alpha^{n-j}_{k-l} \\
        \overline{P}(j,k) &=\overline{P}(k,j) = \overline{P}(n-j,k)=\overline{P}(j,n-k)\nonumber.
    \end{align}
    This result gives a new method for running the Burnside process. Equation~\ref{eq:C2-lumped} says that the distribution $\overline{P}(j,\cdot)$ is the convolution of two discrete arcsine distributions. Thus, to transition from $j$ to $k$ we can independently sample $k_1$ and $k_2$ from the discrete arcsine distributions with parameters $j$ and $n-j$ and then set $k=k_1+k_2$. It is easy to sample from the discrete arcsine distribution since the discrete arcsine distribution is a Beta-Binomial distribution with parameters $n$ and $\alpha=\beta=1/2$.
\end{example}

\subsection{Rates of convergence for the Burnside process}

It is natural to ask about rates of convergence for the Burnside process. \citet{goldberg2002burnside} shows that the Burnside process can be exponentially slow for certain examples. In Section~\ref{sec:transposing} we introduce a variation on the Burnside process for partitions and find empirically that it converges incredibly rapidly. For P\'olya trees as described in Example~\ref{ex:trees}, \citet{bartholdi2024algorithm} report empirical convergence in 30 steps for $n=10,000$ (!). In general there are strong connections between the Burnside process and the popular Swendsen--Wang algorithm~\citep{jerrum1993uniform,diaconis_hnr}. 

There has been some work on rates of convergence in special cases. For $\X = C_2^n$ and $G=S_n$ as in Examples~\ref{ex:C2} and \ref{ex:C2-revisited}, \citet{diaconis2020hahn} show that
\begin{theorem}[Theorem~1.1 in \citet{diaconis2020hahn}]
    Let $x_0\in \X=C_2^n$ be the point $(0,0,\ldots,0)$, and $P_{x_0}^j$ be the probability distribution on $\X$ induced by starting at $x_0$ and taking $j$ steps of the Burnside process. Then, for all $j \ge 4$ and all $n$
    \[
        \frac{1}{4}\left(\frac{1}{4}\right)^j \le \Vert P_{x_0}^j- \pi \Vert_{\mathrm{TV}} \le 4 \left(\frac{1}{4}\right)^j.
    \]
\end{theorem}
This shows that a finite number of steps suffices for convergence no matter how large $n$ is. The proof uses an explicit diagonizational of the lumped chain $\overline{P}$ where all the eigenvalues and eigenvectors are known. For the analogous chain on $C_k^n$, Aldous \citep[Section~12.1.6]{aldous2014reversible} gives a coupling argument that shows that order $k\log n$ steps are sufficient. Zhong \citep{zhong2025ewens} studies a weighted Burnside process on $C_k^n$ and provides a sharper analysis.

The Burnside process for contingency tables is studied in Chapter~5 of \citet{dittmer2019counting} where he shows that the Burnside process mixes quickly for large, sparse tables. We also study this example in Section~\ref{sec:tables}. Our emphasis is on implementing the lumped Burnside process and on the connection between contingency tables and double cosets. Our results thus nicely complement \citet{dittmer2019counting}. 

Careful analysis of a Burnside process for set partitions in is \citet{paguyo2022mixing}. Mixing time results for a nice collection of conjugacy class chains is in \citet{rahmani2022mixing}. The results in \citet{rahmani2022mixing} work for `CA groups.' A group is a CA group if the centralizer of any non-central element is Abelian. These include the three-dimensional Heisenberg group mod $p$, the affine group mod $p$ and $\mathrm{GL}_2(\mathbb{F}_p)$. For these groups, a combination of character theory, coupling and geometric arguments give matching upper and lower bounds. 

All of this said, careful mixing time analysis of the Burnside process is in an embryonic state. This includes the two examples below. We are working on it.

\section{Contingency tables}\label{sec:tables}

A \emph{contingency table} is an $I \times J$ array of non-negative integers $(T_{i,j})$ frequently used to record cross classified data. For example, Tables~\ref{fig:table1} and \ref{fig:table2} show two classical data sets; the first is a classification by hair color and eye color. The second is a classification by number of children and yearly income. A standard task of statistical analysis of such cross classified data is to test for independence of the two variables. This test is most often done by computing the chi-square statistic
\begin{equation}
    \chi^2 =\chi^2(T)= \sum_{i,j} \frac{(T_{i,j}-r_ic_j/n)^2}{r_{i}c_{j}/n}, \label{eq:chi-sq}
\end{equation}
with $r_{i}=\sum_{j}T_{i,j}$, $c_{j}=\sum_{i}T_{i,j}$ and $n=\sum_{i,j}T_{i,j}$. Standard practice compares $\chi^2(T)$ to its limit distribution--the chi-squared distribution on $(I-1)(J-1)$ degrees of freedom. To illustrate, in Table~\ref{fig:table1}, $I=J=4$, $\chi^2=138.29$ on $3 \times 3=9$ degrees of freedom. The value $138.29$ is way out in the tail of the limiting distribution and the standard test rejects the null hypothesis of independence.

\begin{table*}[t]
    \centering
    \caption{Eye Color vs hair color for $n=592$ subjects. Reproduced from Table~1 in \citet{diaconis1985volume}}
    \begin{tabular}{crrrrr}
    \toprule
    & \multicolumn{4}{c}{\textbf{Hair Color}} & \\ 
    \cmidrule(lr){2-5}
    \textbf{Eye Color} & \textbf{Black} & \textbf{Brunette} & \textbf{Red} & \textbf{Blond} & \textbf{Total} \\
    \midrule
    Brown & 68 & 119 & 26 & 7 & 220 \\
    Blue  & 20 & 84  & 17 & 94 & 215 \\
    Hazel & 15 & 54  & 14 & 10 & 93 \\
    Green & 5  & 29  & 14 & 16 & 64 \\
    \midrule
    \textbf{Total} & 108 & 286 & 71 & 127 & 592 \\
    \bottomrule
    \end{tabular}
    \label{fig:table1}
\end{table*}

\begin{table*}[t]
    \centering
    \caption{Number of children vs yearly income for $n=25,263$ Swedish families. Reproduced from Table~2 in \citet{diaconis1985volume}}
    \begin{tabular}{crrrrr}
    \toprule
    &\multicolumn{4}{c}{\textbf{Yearly Income}}& \\ 
    &\multicolumn{4}{c}{\textbf{Units of 1,000 Kroner}}&\\
    \cmidrule(lr){2-5}
    \textbf{Number of Children} & \textbf{0--1} & \textbf{1--2} & \textbf{2--3} & \textbf{3+} & \textbf{Total} \\
    \midrule
    0  & 2161 & 3577 & 2184 & 1636 & 9558 \\
    1  & 2755 & 5081 & 2222 & 1052 & 11110 \\
    2  & 936  & 1753 & 640  & 306  & 3635  \\
    3  & 225  & 419  & 96   & 38   & 778   \\
    $\geq 4$ & 39   & 98   & 31   & 14   & 182   \\
    \midrule
    \textbf{Total} & 6116 & 10928 & 5173 & 3046 & 25263 \\
    \bottomrule
    \end{tabular}
    \label{fig:table2}
\end{table*}

The statistical background is covered in any introductory statistical text and \citet{agresti1992survey} is an extensive survey. Early on, R.A. Fisher \citep{fisher1970statistical} suggested looking at the conditional distribution of the chi-square statistic when the row and column sums are fixed, and the two variables are independent. In a later development, Diaconis and Efron \citep{diaconis1985volume} suggested looking at the uniform distribution on tables with fixed row and columns sums as a way of interpreting `large' values of $\chi^2$. 

\citet{diaconis1985volume} led to a healthy development, now subsumed under the label `algebraic statistics.' The main questions are can one enumerate all such tables or pick one at random. These problems turn out to be $\#P$-complete, meaning exact enumeration is impossible and a host of approximations have been developed \citep{dyer1997sampling}. An influential paper of Diaconis and Sturmfels \citep{algebraic} develops Markov chain Monte Carlo approximations. It also contains a comprehensive review of conditional testing. \citet{SIS} further treat these contingency table problems using sequential importance sampling. See \citet{almendra2024markov} for recent references.

This section develops an algorithm for uniform generation of contingency tables with given row and column sums. The algorithm uses the Burnside process for double cosets in the permutation group $S_n$ and previously appeared in \citet[Chapter~5]{dittmer2019counting}. Double cosets are reviewed in Section~\ref{sec:double-cosets} which also contains a general procedure for running the Burnside process for double cosets. The connection between contingency tables and double cosets in $S_n$ is discussed in Section~\ref{sec:tables-cosets}. In Section~\ref{sec:tables-burnside} we describe the lumped Burnside process for contingency tables and in Section~\ref{sec:tables_complexity} we analyze the complexity of both the original and lumped processes. The lumped process is applied to Tables~\ref{fig:table1} and \ref{fig:table2} in Section~\ref{sec:tables-examples}. We compare the results with previous efforts \citep{diaconis1985volume,algebraic} and summarize our experience.

\subsection{Double cosets}\label{sec:double-cosets}

In this section we develop a general method for uniformly sampling double cosets via the Burnside process. The connection to contingency tables is developed in Sections~\ref{sec:tables-cosets} and \ref{sec:tables-burnside}. Throughout $H$ and $K$ will be subgroups of a finite group $G$. The product group $H \times K$ acts on $G$ by
\begin{equation}
    s^{h,k} = h^{-1}sk.\label{eq:HK-action}
\end{equation}
This gives an equivalence relation on $G$
\[
    s \sim t \Longleftrightarrow t = h^{-1}sk \text{ for some } (h,k) \in H \times K.
\]
The equivalence classes are double cosets, $HsK$ denotes the double coset containing $s$ and  $H\backslash G \slash K$ denotes the set of all double cosets. These definitions are standard in basic group theory \cite[Page~23]{suzuki1986group}. Classical facts are
\begin{align}
    |HsK| &= \frac{|H||K|}{|H\cap sKs^{-1}|},\label{eq:coset-size}\\
    |H\backslash G\slash K| &=\frac{1}{|H||K|}\sum_{h,k}|G^{h,k}|,\nonumber
\end{align}
where $G^{h,k} = \{s : s^{h,k} = s\}$. Despite these nice formulas, enumerating or describing double cosets can be an intractable problem. \citet{diaconis2022statistical} study and provide references for a host of examples where the enumeration problems are interesting. Many of the examples have connections to areas of mathematics beyond group theory. 

The Burnside process applies to the action of $H\times K$ on $G$ and generates uniformly chosen double cosets. In this case, the Burnside process is a Markov chain on $G$. To carry it out, two steps must be implemented:
\begin{enumerate}
    \item From $s\in G$, uniformly choose $(h,k)$ in 
    \[
        (H \times K)_s = \{(h,k)\in H \times K:s^{h,k}=s\}.
    \]
    \item From $(h,k)$, uniformly choose $t$ in 
    \[
        G^{h,k}=\{t \in G : t^{h,k}=t\}.
    \]
\end{enumerate}
To implement these steps, here is a characterization of $(H \times K)_s$ and $G^{h,k}$. In the following, $C_G(h)$ denotes the centralizer of $h$ in $G$. That is $C_G(h) = \{s \in G: sh = hs\}$.
\begin{lemma}\label{lem:cosets}
    Let $H \times K$ act on $G$ as in \eqref{eq:HK-action}. Then
    \begin{enumerate}
        \item For all $s \in G, h \in H$ and $k \in K$, $s^{h,k}=s$ if and only  if $k=s^{-1}hs$.
        \item For all $s \in G$, 
        \[
            (H \times K)_s = \{(h,s^{-1}hs) : h \in H \cap sKs^{-1}\}.
        \]
        \item For all $(h,k) \in H \times K$, if $G^{h,k}$ is non-empty, then $G^{h,k}=C_G(h)s$ where $s$ is any fixed element of $G^{h,k}$.
    \end{enumerate}
\end{lemma}
Lemma~\ref{lem:cosets} is proved in Appendix~\ref{appn:cosets}. The lemma implies that following procedure is equivalent to the Burnside process for $H \times K$ acting on $G$:
\begin{enumerate}
    \item From $s \in G$, choose $h$ uniformly from $H \cap sKs^{-1}$.
    \item From $h$, choose $g$ uniformly in $C_G(h)$ and set $t=gs$.
\end{enumerate}
Thus, to run the Burnside process for double cosets it suffices to be able to sample from the subgroups $H \cap sKs^{-1}$ and $C_G(h)$ for all $s \in G$ and $h \in H$. The above description will be useful in Section~\ref{sec:tables-burnside} when we derive the lumped Burnside process for contingency tables.

\subsection{Contingency tables as double cosets}\label{sec:tables-cosets}

This section explains the relationship between contingency tables with given row and column sums and double cosets in $S_n$. We roughly follow Section~1.3 of \citet{james1981representation}. Throughout $\lambda=(\lambda_i)_{i=1}^I$ and $\mu=(\mu_j)_{j=1}^J$ are two compositions of $n$. That is, $\lambda_i$ and $\mu_j$ are positive integers with $\lambda_1+\cdots+\lambda_I=\mu_1+\cdots+\mu_J=n$. Since $\lambda$ and $\mu$ are compositions of $n$ rather than partitions, we do not assume that they are in non-decreasing order. 

The space of contingency tables with row sums $\lambda$ and column sums $\mu$ will be denoted by $\T_{\lambda,\mu}$. That is, $\T_{\lambda,\mu}$ is the set
\begin{equation*}
    \left\{T \in \integers^{I \times J}_{\ge 0} : r_{i} = \lambda_i \text{ and } c_j=\mu_j \text{ for all } i, j\right\},
\end{equation*}
where, as before, $r_{i}=\sum_{j} T_{i,j}$ and $c_{j}=\sum_i T_{i,j}$. The composition $\lambda$ defines a set partition $L=(L_i)_{i=1}^I$ with 
\begin{align*}
    L_1&=\{1,\ldots,\lambda_1\},\\
    L_2&=\{\lambda_1+1, \ldots,\lambda_1+\lambda_2\}\\
    \vdots\\
    L_I&=\{n-\lambda_I+1,\ldots,n\}.
\end{align*} 
Likewise, $\mu$ defines the set partition $(M_j)_{j=1}^J$ with  
\begin{align*}
    M_1&=\{1,\ldots,\mu_1\}\\
    M_2&=\{\mu_1+1,\ldots,\mu_1+\mu_2\}\\
    \vdots \\
    M_J&=\{n-\mu_J+1,\ldots,n\}.
\end{align*}
For any set partition $A=(A_i)_{i=1}^I$, define the subgroup $S_A \subseteq S_n$ by
\begin{equation}
    \label{eq:S_A} S_A = \{\sigma \in S_n : \sigma(A_i)=A_i \text{ for all } i\}.
\end{equation}
Permutations in $S_A$ permute the elements of $A_1$ amongst themselves, the elements of $A_2$ amongst themselves and so on. Thus, $S_A \cong \prod_{i=1}^I S_{|A_i|}$. The subgroup $S_A$ is called a \emph{Young subgroup} or a \emph{parabolic subgroup}. For the compositions $\lambda$ and $\mu$, we will use $S_\lambda$ and $S_\mu$ to denote $S_L$ and $S_M$. The following map will be used to define a bijection between $S_\lambda \backslash S_n \slash S_\mu$ and $\T_{\lambda,\mu}$. 
\begin{definition}\label{def:biject}
    Let $(L_i)_{i=1}^I$ and $(M_j)_{j=1}^J$ be the set partitions above. Define a function $f:S_n \to \integers_{\ge 0}^{I \times J}$ by
    \[
        f(\sigma)_{i,j} = |L_i \cap \sigma(M_j)|.
    \]
    That is, $f(\sigma)$ is a contingency table and the $(i,j)$-th entry of $f(\sigma)$ is the number of elements in the set $M_j$ that $\sigma$ maps into the set $L_i$. 
\end{definition}
The image of $f$ is the space of contingency tables $\mathcal{T}_{\lambda,\mu}$, and it can be shown that $f(\sigma)=f(\tau)$ if and only if $S_\lambda\sigma S_\mu=S_\lambda\tau S_\mu$. In particular, we have the following result.

\begin{proposition}[Theorem~1.3.10 in \citet{james1981representation}]\label{prop:biject}
    The function $f$ in Definition~\ref{def:biject} is constant on the double cosets in $S_\lambda \backslash S_n \slash S_\mu$ and induces a bijection between $S_\lambda \backslash S_n \slash S_\mu$ and $\T_{\lambda,\mu}$.
\end{proposition}

Sampling $\sigma \in S_n$ uniformly and then computing $f(\sigma)$ induces a probability measure on $\T_{\lambda,\mu}$. This probability distribution is called the \emph{Fisher---Yates distribution} or the multiple hypergeometric distribution. Under the Fisher--Yates distribution, the probability of a contingency table $T=f(\sigma)$ is proportional to  $|S_\lambda \sigma S_\mu|$ which can be computed using \eqref{eq:coset-size}. Note that
\[
 S_\lambda \cap \sigma S_\mu \sigma^{-1} = S_L \cap S_{\sigma M}= S_{L \land \sigma M}, 
\]
where $\sigma M$ is the set partition $(\sigma(M_j))_{j=1}^J$ and $L \land \sigma M$ is the \emph{meet} of $L$ and $\sigma M$ meaning
\[
    L \land \sigma M = (L_i \cap \sigma(M_j))_{i=1,j=1}^{I,J}.
\]
The subgroup $S_\lambda \cap \sigma S_\mu \sigma^{-1}$ therefore has size $\prod_{i,j}f(\sigma)_{i,j}!$. These observations imply the following.
\begin{proposition}
    Let $p$ be the probability mass function for the Fisher--Yates distribution on $\T_{\lambda,\mu}$, then for all $T \in \T_{\lambda,\mu}$,
    \begin{equation}
        p(T)=\frac{1}{n!} \frac{\left(\prod_{i=1}^I \lambda_i!\right)\left(\prod_{j=1}^J \mu_j!\right)}{\prod_{i,j}T_{i,j}!}.\label{eq:FY}
    \end{equation}
\end{proposition} 
The Fisher--Yates distribution plays a central role in our implementation of the lumped Burnside process for contingency tables.

\subsection{The lumped process for contingency tables}\label{sec:tables-burnside}

In this section we derive the lumped Burnside process for contingency tables. That is, we describe the transitions from $f(\sigma) \in \T_{\lambda,\mu}$ to $f(\tau)$ when $\sigma \in S_n$ transitions to $\tau$ according to the Burnside process for $S_\lambda,S_\mu$ double cosets. In Section~\ref{sec:tables_complexity} we will quantify the computational benefits of the lumped process over the original process. These benefits are also demonstrated empirically in Section~\ref{sec:tables-examples} using the two examples in Tables~\ref{fig:table1} and \ref{fig:table2}. A similar description of the lumped Burnside process has previously appeared in \citet{dittmer2019counting}. However, his implementation has complexity that scales linearly in $n$, the sum of the table entries. As shown later, the average case complexity of our algorithm is on the order of $(\log n)^2$.

Lemma~\ref{lem:lumped_tables} contains the key insights behind our implementation of the lumped Burnside process. It describes the distribution of $f(\tau) \in \T_{\lambda,\mu}$ when $\tau$ is drawn uniformly from the set $S_n^{h,k}$ as in the second step of the Burnside process. In Corollary~\ref{cor:lumped_tables}, we observe that the distribution of $f(\tau)$ only depends on the cycle type of $h$ and $k$. Lemma~\ref{lem:lumped_tables2} then describes the distribution of the cycle type of $h$ and $k$ when $(h,k)$ is uniformly sampled from $(S_\lambda \times S_\mu)_\sigma$ as in the first step of the Burnside process.

Before stating these lemmas, we will set some notation and states two definitions. Through-out this section,  $(L_i)_{i=1}^I$ and $(M_j)_{j=1}^J$ are the set partitions corresponding to $\lambda$ and $\mu$. $h \in S_\lambda$ and $k \in S_\mu$ can be represented as $h=(h_i)_{i=1}^I$ and $k=(k_j)_{j=1}^J$ where $h_i$ and $k_j$ are permutations of $L_i$ and $M_j$ respectively. For any permutation, $h$, we will let $\mathcal{C}(h)$ denote the set of cycles of $h$.

\begin{definition}\label{def:lumped_tables1}
    Fix permutations $h=(h_i)_{i=1}^I \in S_\lambda$ and $k=(k_j)_{j=1}^J \in S_\mu$. For each $l \ge 1$, define $r^{(l)} \in \integers_{\ge 0}^I$ and $c^{(l)} \in \integers_{\ge 0}^J$ by
    \begin{align*}
        r^{(l)}_i& = |\{C \in \mathcal{C}(h_i): |C|=l\}|, \text{ and}\\
         c_j^{(l)}&=|\{C \in \mathcal{C}(k_j) : |C|=l\}|.
    \end{align*}
    In words, $r^{(l)}_i$ is the number of $l$ cycles of $h_i$ and $c^{(l)}_j$ is the number of $l$ cycles of $k_j$. The vector $(r^{(l)}_i)_{l \ge 1}$ is called the \emph{cycle type} of $h_i$ (and likewise for $(c^{(l)}_i)_{l \ge 1}$).
\end{definition}

\begin{definition}\label{def:lumped_tables2}
    Fix permutations $h=(h_i)_{i=1}^I \in S_\lambda$ and $k=(k_j)_{j=1}^J \in S_\mu$. Let  $\tau$ be a permutation in the set of fixed points $S_n^{h,k}$. For each $l \ge 1$, define $X^{(l)} \in \integers_{\ge 0}^{I \times J}$ by
    \[
        X^{(l)}_{i,j} = |\{C : |C| = l, C \in \mathcal{C}(h_i), \tau^{-1}(C) \in \mathcal{C}(k_j)\}|.
    \]
    In words, $X_{i,j}^{(l)}$ is the number of $l$ cycles of $h_i$ that $\tau^{-1}$ maps to an $l$ cycle of $k_j$. 
\end{definition}

Lemma~\ref{lem:lumped_tables} expressed the contingency table $f(\tau)$ in terms of the tables $X^{(l)}$ and describes the distribution of $X^{(l)}$ when $\tau$ is drawn uniformly from $S_n^{h,k}$.

\begin{lemma}\label{lem:lumped_tables}
    Fix $h  \in S_\lambda$ and $k\in S_\mu$ and let $r^{(l)}$ and $c^{(l)}$ be as in Definition~\ref{def:lumped_tables1}. If $\tau$ is $S_{n}^{h,k}$ and $X^{(l)}$ is as in Definition~\ref{def:lumped_tables2}, then
    \[
        f(\tau)=\sum_{l=1}^n lX^{(l)}.
    \]
    Furthermore, if $\tau$ is drawn uniformly from $S_n^{h,k}$, then the table $X^{(l)}$ is distributed according to the Fisher--Yates distribution with row sums $r^{(l)}$ and $c^{(l)}$ and $(X^{(l)})_{l \ge 1}$ are independent.
\end{lemma}

Lemma~\ref{lem:lumped_tables} is proved in Appendix~\ref{appn:lumped_tables}. To get a feel for the lemma, the following examples are helpful.

\begin{example}
    Suppose that $h$ and $k$ are both the identity permutation. In this case $r^{(1)}_i =\lambda_i$, $c^{(1)}_j = \mu_j$ and $r^{(l)}_i=c^{(l)}_j=0$ otherwise. We also have that $S_n^{h,k} = S_n$. Thus, Lemma~\ref{lem:lumped_tables} is simply stating that if $\tau$ is uniformly distributed in $S_n$, then $f(\tau)$ is distributed according to the Fisher--Yates distribution.
\end{example}

\begin{example}
    Suppose that $\lambda = \mu = (m,\ldots,m)$ ($I$ times) and that for each $i$, $h_i=k_i$ is an $m$-cycle. Then $r^{(m)}_i = c^{(m)}_i=1$ and $r^{(l)}_i=c^{(l)}_i=0$ for $l \neq m$. In this case Lemma~\ref{lem:lumped_tables} is stating that $f(\tau) = mX^{(m)}$ and that $X^{(m)} \in \integers_{\ge 0}^{I \times I}$ is a uniformly drawn permutation matrix. 
\end{example}

Lemma~\ref{lem:lumped_tables} has the following corollary which is useful for the lumped Burnside process.
\begin{corollary}\label{cor:lumped_tables}
    Let $h$, $k$ be as in Lemma~\ref{lem:lumped_tables}. If $\tau$ is uniformly distributed in $S_n^{h,k}$, then the distribution of $f(\tau)$ only depends on the vectors $(r^{(l)})_{l \ge 1}$ and $(c^{(l)})_{l \ge 1}$ from Definition~\ref{def:lumped_tables3}. 
\end{corollary}

Corollary~\ref{cor:lumped_tables} states that, when running the lumped Burnside process, we do not need to sample the full permutations $h$  and $k$. It is sufficient to just sample $r_i^{(l)}$ and $c_j^{(l)}$ from the correct distribution. Definition~\ref{def:lumped_tables3} and Lemma~\ref{lem:lumped_tables2} describe this distribution.
\begin{definition}\label{def:lumped_tables3}
    Let $\sigma$ be a permutation in $S_n$ and let $(h,k)$ be an element of the stabilizer $(S_\lambda \times S_\mu)_{\sigma}$. For each $i \in [I], j \in [J]$ and $l \ge 1$, define $a_{i,j}^{(l)}$ by 
    \[
        a_{i,j}^{(l)}=|\{C: |C|=l, C \in \mathcal{C}(h), C \subseteq L_i \cap \sigma(M_j)\}|.
    \]
    That is, $a_{i,j}^{(l)}$ is the number of $l$ cycles of $h$ contained in $L_i \cap \sigma(M_j)$.
\end{definition}
Lemma~\ref{lem:lumped_tables2} expresses $r_i^{(l)}$ and $c_j^{(l)}$ in terms of $a_{i,j}^{(l)}$ and describes the distribution of $a_{i,j}^{(l)}$.
\begin{lemma}\label{lem:lumped_tables2}
    Fix $T \in \T_{\lambda,\mu}$ and let $\sigma$ be any permutation with $f(\sigma)=T$. Let $(h,k)$ be uniformly sampled from $(S_\lambda \times S_\mu)_\sigma$. Let $r^{(l)}_i$ and $c^{(l)}_j$ be as in Definition~\ref{def:lumped_tables2} and let $a_{i,j}^{(l)}$ be as in Definition~\ref{def:lumped_tables3}. Then, 
    \[
        r^{(l)}_i = \sum_{j=1}^J a_{i,j}^{(l)} \quad \text{and}\quad c^{(l)}_j = \sum_{i=1}^{I}a_{i,j}^{(l)}.
    \]
    Furthermore, $a_{i,j}=\left(a_{i,j}^{(l)}\right)_{l \ge 1}$ is equal in distribution to the cycle type of a uniform permutation of length $T_{i,j}$ and $(a_{i,j})_{i,j}$ are independent across $i$ and $j$.
\end{lemma}

The proof of Lemma~\ref{lem:lumped_tables2} is given in Appendix~\ref{appn:lumped_tables2}. In our implementation, $(a_{i,j}^{(l)})$ are sampled efficiently using discrete stick breaking as described in Definition~\ref{def:stick-breaking}. The two steps of the Burnside process for contingency tables are
\begin{enumerate}
    \item From a contingency table $T$, sample $(a_{i,j}^{(l)})$ using discrete stick breaking and compute $(r^{(l)})_{l \ge 1}$ and $(c^{(l)})_{l \ge 1}$ as in Lemma~\ref{lem:lumped_tables2}.
    \item Given $(r^{(l)})_{l \ge 1}$ and $(c^{(l)})_{l \ge 1}$ sample $X^{(l)}$ and return $T' = \sum_{l \ge 1}lX^{(l)}$ as in Lemma~\ref{lem:lumped_tables}.
\end{enumerate}

A full description of the lumped chain is given at the end of this section in Algorithm~\ref{alg:tables}. The example below walks through a single step of the lumped Burnside process started from a contingency table of size 2 by 2.

\begin{example}
    Let $I=J=2$ and $\lambda=\mu=(5,5)$ (and so $n=10$) Suppose that the current state of the lumped Burnside process is the table 
    \[  T=\begin{bmatrix}
        5,0\\0,5
        \end{bmatrix}.
    \] In the first step of Algorithm~\ref{alg:tables}, we apply discrete stick breaking to the non-zero entries of $T$. Suppose that the upper left entry of $T$ splits into three fixed points and one 2 cycle. Suppose also that low right entry of $T$ splits into one fixed point and two 2 cycles. In the notation of Definition~\ref{def:lumped_tables3}, we have
    \[
    \begin{bmatrix}
        a_{1,1}^{(1)}, a_{1,2}^{(1)}\\
        a_{2,1}^{(1)}, a_{2,2}^{(1)}\\
    \end{bmatrix}=\begin{bmatrix}
        3&0\\
        0&1
    \end{bmatrix},\text{ and }  \begin{bmatrix}
        a_{1,1}^{(2)}, a_{1,2}^{(2)}\\
        a_{2,1}^{(2)}, a_{2,2}^{(2)}\\
    \end{bmatrix}=\begin{bmatrix}
        1&0\\
        0&2
    \end{bmatrix},
    \]
    and $a_{i,j}^{(l)}=0$ for $l> 2$. With these values of $a_{i,j}^{(l)}$, we get $r^{(1)}=c^{(1)}=[3,1]$ and $r^{(2)}=c^{(2)}=[1,2]$. In the second step of Algorithm~\ref{alg:tables}, $X^{(l)}$ is sampled from the Fisher--Yates distribution. Suppose that we sample 
     \[
    X^{(1)}=\begin{bmatrix}
        2&1\\
        1&0
    \end{bmatrix},\text{ and }  X^{(2)}=\begin{bmatrix}
        0&1\\
        1&1
    \end{bmatrix}.
    \]
    Then, the new state of the lumped Burnside process will be
    \[
    T' = X^{(1)} + 2 X^{(2)} = \begin{bmatrix}
        2&3\\
        3&2
    \end{bmatrix} \in \mathcal{T}_{\lambda,\mu}.
    \]
\end{example}

\begin{algorithm*}
\caption{Lumped Burnside process for contingency tables}
\begin{algorithmic}[1]
\Require{$T$ (contingency table of size $I \times J$)}
\Ensure{$T'$ (new contingency table sampled from the Burnside process)}

\Comment{Step 1}

\State{$r_{i}^{(l)} \gets 0$ for $i \in [I]$  and $l \in [n]$}
\State{$c_{j}^{(l)} \gets 0$ for $j \in [J]$ and $l \in [n]$}
\For{$i \gets 1$ to $I$}
    \For{$j \gets 1$ to $J$}
      \State{Sample $(\lambda_m)_{m=1}^M$ from the discrete stick breaking distribution on $T_{i,j}$}
      \For{$m \gets 1$ to $M$}
        \State{$r_{i}^{(\lambda_m)} \gets r_i^{(\lambda_m)} + 1$}
        \State{$c_{j}^{(\lambda_m)} \gets c_j^{(\lambda_m)}+ 1$}
      \EndFor  
    \EndFor
\EndFor

\Comment{Step 2}

\State{$T'_{i,j} \gets 0$ for $i \in [I]$ and $j \in [J]$}
\For{$l$ such that $r^{(l)}\neq 0$}
    \State{Sample $X^{(l)}$ from the Fisher--Yates distribution with row sums $r^{(l)}$ and column sums $c^{(l)}$}
    \State{$T' \gets T' + lX^{(l)}$}
\EndFor
\State \Return{$T'$}
\end{algorithmic}
\label{alg:tables}
\end{algorithm*}

\subsection{Computational complexity of the lumped Burnside process}\label{sec:tables_complexity}

The lumped chain has both memory and speed advantages over the original Burnside process. The memory advantages are easier to see. To run the original process, we need to store the permutation $\sigma$ which has memory requirement of the order $n \log n$ (since we are storing $n$ integers with order $\log n$ digits).  For the lumped process, we need to store the table $T$ which requires on the order of $IJ \log n$ memory.  The lumped chain also requires storing the non-zero values of $r_{i}^{(l)}$ and $c_j^{(l)}$. The number of non-zero values of $r_{i}^{(l)}$ and $c_j^{(l)}$ is upper bounded by the number of non-zero values of $a_{i,j}^{(l)}$. Since $a_{i,j}^{(l)}$ is drawn from a discrete stick breaking distribution, there are on average $O(\log T_{i,j})$ values of $l$ with $a_{i,j}^{(l)} \neq 0$. Thus, the total memory requirement of the lumped chain is $O(IJ \log n)$ on average.

The speed advantages of the lumped chain are also substantial. The values $(a_{i,j}^{(l)})_{l \ge 0}$ can be sampled using stick breaking with an average case complexity on the order of $(\log T_{i,j})^2$. Therefore, $(r^{(l)})_{l \ge 1}$ and $(c^{(l)})_{l \ge 1}$ can all be computed with complexity $O(IJ (\log n)^2)$. Given $(r^{(l)})_{l \ge 1}$ and $(c^{(l)})_{l \ge 1}$, we only need to sample $X^{(l)}$ for the non-zero values of $r^{(l)}$. Thus, for all but $IJ\log n$ values of $l$, $X^{(l)}$ is zero and does not need to be computed. For the non-zero values of $r^{(l)}$, $X^{(l)}$ can be sampled with complexity $O(IJ\log n)$. Thus, the time complexity of the lumped chain is $O((IJ\log n)^2)$. This is a substantial speed up over the original chain which has linear complexity. Indeed, simply computing $f(\sigma)$ has complexity of order $n$.

\subsection{Applications of the lumped process for contingency tables}\label{sec:tables-examples}

In this section, we run the Burnside process on the spaces of contingency tables corresponding to Tables~\ref{fig:table1} and \ref{fig:table2}. Our simulations agree with the results reported in \citet{diaconis1985volume} and highlight the benefit of the lumped process. 

Both Tables~\ref{fig:table1} and \ref{fig:table2} have values of the chi-square statistic \eqref{eq:chi-sq} that are very large compared to the chi-square statistics of tables drawn from the Fisher--Yates distribution \eqref{eq:FY}. In such settings, \citet{diaconis1985volume} suggest comparing the chi-square statistic to the statistics of tables drawn from the uniform distribution on $\mathcal{T}_{\lambda,\mu}$. Specifically, they suggest computing the \emph{volume statistic} given by
\[
    V(T) = \frac{\left|\{T' \in \mathcal{T}_{\lambda,\mu} : \chi^2(T') \ge \chi^2(T) \}\right|}{\left|\mathcal{T}_{\lambda,\mu}\right|}.
\]
That is, $V(T)$ is the proportion of tables $T' \in \mathcal{T}_{\lambda,\mu}$ that have a larger chi-square statistic than $T$. Here $\lambda$ and $\mu$ are the row and column sums of $T$. Large values of $V(T)$ are evidence of stronger dependence between the row and column variables in $T$. Note that $V(T)$ can also be written as
\begin{equation}
    V(T) = \Prob(\chi^2(T') \ge \chi^2(T)),\label{eq:vol_test}
\end{equation}
where $T'$ is drawn from the uniform distribution on $\mathcal{T}_{\lambda,\mu}$. Thus, $V(T)$ can be approximated by running the Burnside process on $\mathcal{T}_{\lambda,\mu}$ and reporting the proportion of sampled tables with $\chi^2(T') \ge \chi^2(T)$. The results of such a simulation are reported in Tables~\ref{fig:table-sim}. The simulation agrees with the findings in \citet{diaconis1985volume}. They find that Table~\ref{fig:table2} has a value of $\chi^2(T)$ that is much smaller than expected under the uniform distribution. On the other hand, Table~\ref{fig:table1} has a value of $\chi^2(T)$ that is fairly typical under the uniform distribution.  The numerical value of $V(T)$ for Table~\ref{fig:table1} agrees with the value reported in Section~6 of \citet{diaconis1995rectangular}. They use a different Markov chain to estimate $V(T)$ and report that $V(T) \approx 0.154$.

\begin{table*}
    \centering
    \caption{Estimates of the volume test statistic $V(T)$ for Tables~\ref{fig:table1} and \ref{fig:table2}. Each estimate was produced by running the lumped Burnside process for $2\times 10^6$ steps. The Burnside process was initialized at either Table~\ref{fig:table1} or Table~\ref{fig:table2} and run for $10^4$ steps as a `burn-in.' For each table, this procedure is repeated five times to get five estimates of $V(T)$. These five estimates and their median are reported for each table. On a personal computer, 1 hour was needed to perform all $10^7+5\times 10^4$ steps for Table~\ref{fig:table1}. For Table~\ref{fig:table2}, 4 hours and 20 minutes were needed to perform all $10^7+5\times 10^4$ steps.}
    \begin{tabular}{c|ccccc|c}
        & Run 1 & Run 2 & Run 3 & Run 4 & Run 5&Median\\
        \hline
        Table~\ref{fig:table1}&0.1545&0.1534&0.1532&0.1535&0.1533&0.1534\\
        Table~\ref{fig:table2}&$1.6\times 10^{-5}$&$1.35\times 10^{-5}$&$6.0\times 10^{-6}$&$8.0\times 10^{-6}$&$1.35\times 10^{-5}$&$1.35\times 10^{-5} $\\
    \end{tabular}
    \label{fig:table-sim}
\end{table*}

There are many methods that use Markov chain Monte Carlo in the context of hypothesis testing to estimate a tail probability such as \eqref{eq:vol_test}. The method we use to estimate $V(T)$ in Table~\ref{fig:table-sim} is consistent meaning that the estimates converge to $V(T)$ as the number of samples goes to infinity. However, the method does not have the error control guarantee that is sometimes desired for hypothesis testing. The methods in \citet{BC,Gerry1,Gerry2,howes2023markov} all provide such a guarantee. The remarkable thing about these methods is that they assume nothing about the mixing time of the Markov chain. Any of these methods can be used in conjunction with the lumped Burnside process.

Finally, Figure~\ref{fig:tables-run-time} shows the computational benefits of the lumped Burnside process.  This figure, and the analysis in Section~\ref{sec:tables_complexity}, show that the lumped process gives an exponential speed up over the unlumped process. The unlumped process has complexity that is linear in $n$ (the sum of the table entries). However, the complexity of the lumped process grows like a power of $\log n$.

\begin{figure*}
    \centering
    \includegraphics[width = 0.8\textwidth]{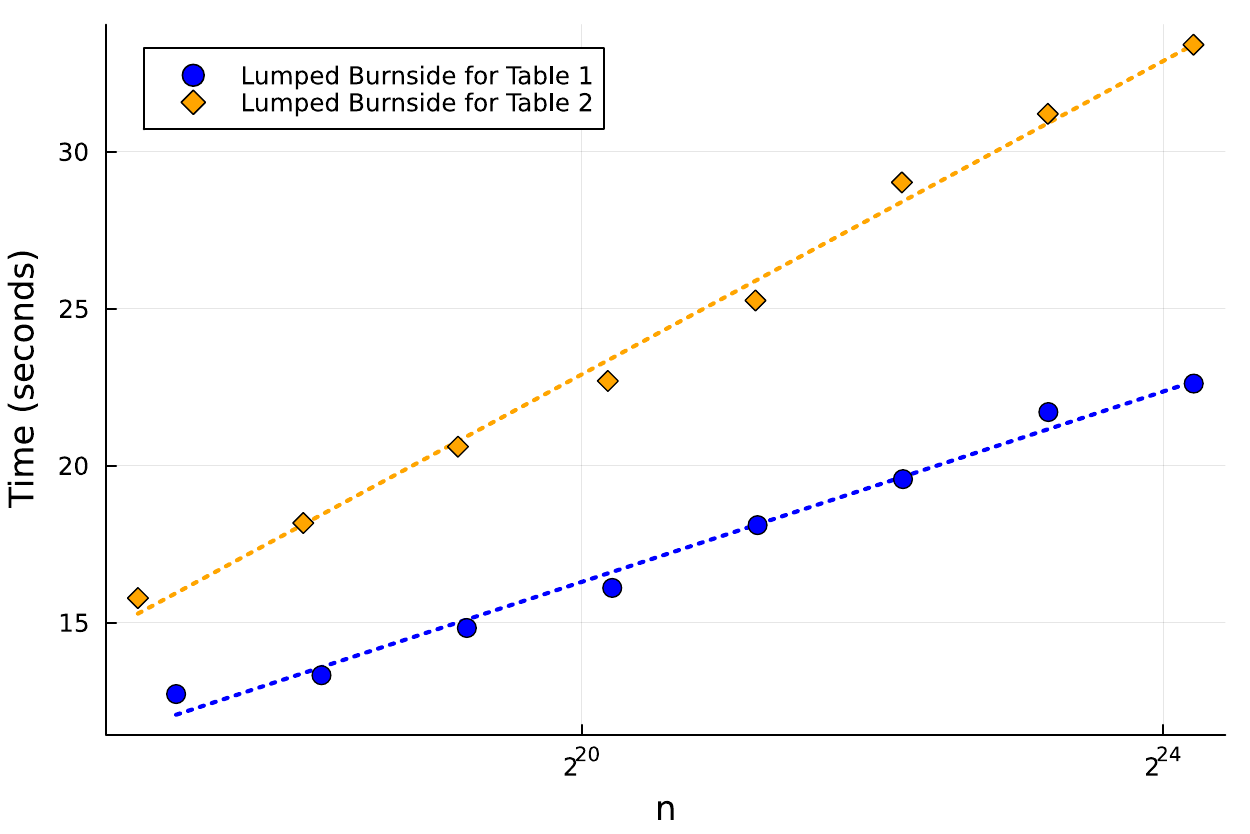}
    \caption{The above figure is a linear--log plot showing the per-step complexity of the Burnside process for contingency tables. The plot was generated by running the Burnside process for 10,000 steps on scaled copies of Tables~\ref{fig:table1} and \ref{fig:table2}. Straight line approximations are shown as dashed lines. We see that if the entries of the table are increased by a multiplicative factor, the run time increases by additive factor. This supports our analysis which shows that the lumped process has a run time that is at most $O((\log n)^2)$ for tables with fixed dimensions.}
    \label{fig:tables-run-time}
\end{figure*}

\section{Partitions}\label{sec:parts}

In combinatorial probability, one studies a set of combinatorial objects (permutations, graphs, matrices) by asking what does a typical element `look like.' The field is rich with limit theorems, and it is natural to ask if these limit theorems are accurate for real world applications (e.g. $n=52$). One approach is to compare the limits to simulations. This comparison in turn requires efficient sampling algorithms. In this section we develop a novel Markov chain for approximately sampling uniform partitions called the \emph{reflected Burnside process}. This Markov chain is used to evaluate limit theorems for partitions.

Section~\ref{sec:parts-background} reviews partitions and some of the limit theorems we compare our simulations to. Section~\ref{sec:lumped-parts} derives the Burnside process for uniformly sampling partitions. We use the fact that partitions are in bijection with the orbits of the permutation group $S_n$ acting on itself by conjugation. We explain this connection and work out the implementation details. Section~\ref{sec:transposing} introduces our reflected Burnside process which we use to illustrate (and test) both the algorithm and the limit theorems in Section~\ref{sec:parts-background}. The reflected Burnside process is an example of speeding up a Markov chain by adding  `deterministic jumps.' Section~\ref{sec:transposing} surveys the literature on this technique. We are excited to have found an example `in the wild.'

\subsection{Uniform partitions}\label{sec:parts-background}

Let $\mathcal{P}_n$ denote the set of partitions of $n$. We will represent partitions in exponential notation as $1^{a_1}2^{a_2}\cdots n^{a_n}$ or $(a_l)_{l= 1}^n$ where $a_l$ is the number of parts of size $l$. This notation is useful for describing the lumped Burnside process in Section~\ref{sec:lumped-parts}. The set of partitions of $n$ is thus
\[
    \mathcal{P}_n = \bigg\{(a_l)_{l = 1}^n :a_l \in \integers_{\ge 0},  \sum_{l= 1}^n la_l = n \bigg\}.
\]
Chapter 10 of \citet{nijenhuis1978combinatorial} contains an algorithm for uniformly sampling from $\mathcal{P}_n$. This algorithm is based on a bijective proof of an identity of Euler's that relates partitions of $n$ to partitions of $m < n$ and divisors of $m-n$. The main limitation of the algorithm in \citet{nijenhuis1978combinatorial} is that it requires computing and storing a table of approximately $n^2$ integers. 

There are also a number of algorithms for sampling partitions that are based on rejection sampling \citep{fristedt1993structure,arratia2016probabilistic}. These methods generate a partition of a random integer $N$ such that conditional on $N=n$, the partition is uniformly distributed on $\mathcal{P}_n$. To get a partition of size $n$, these methods repeatedly sample partitions until one is observed with $N=n$. These methods fit into a broader class of algorithms for sampling from combinatorial structures called \emph{Boltzmann samplers} \citep{Duchon2004Boltzmann,Flajolet2007unlabelled}. The efficiency of the rejection sampling methods depends on the probability of $N=n$. For the method presented in \citet{fristedt1993structure}, this probability is on the order of $n^{-3/4}$. While for the probabilistic divide and conquer method in \citet{arratia2016probabilistic}, the probability that $n=N$ is bounded away from $0$ as $n \to \infty$. The complexity of the method in \citet{arratia2016probabilistic} is $O(n^{1/2+\varepsilon})$ for any $\varepsilon > 0$. This is the same as our conjectured complexity of using the reflected Burnside process in Section~\ref{sec:transposing}. The probabilistic divide and conquer method has the advantage of being an exact algorithm rather than an approximate method like the Burnside process. However, developing the Burnside process in detail here may give insight to other instances of the Burnside process where no exact methods apply. We also believe that it is healthy to have multiple sampling algorithms that can be used to check each other.

\citet{fristedt1993structure} contains many results about the asymptotic distribution of `features' of random partitions as $n\to \infty$. In Section~\ref{sec:transposing}, these asymptotic results are used to empirically evaluate convergence of the reflected Burnside process. We will use the following theorems from \citet{fristedt1993structure}. 
\begin{theorem}[Theorem~2.2 in \citet{fristedt1993structure}]\label{thrm:fristedt1}
    Let $(a_l)_{l =1}^n$ be a uniformly distributed partition of $n$. Then for fixed $l$ and all $x \ge 0$
    \[
       \lim_{n \to \infty} \Prob\left(\frac{\pi}{\sqrt{6n}}la_l \le x\right) = 1-e^{-x}.
    \]
\end{theorem}
\begin{theorem}[Theorem~2.3 in \citet{fristedt1993structure}]\label{thrm:fristedt2}
    Let $(a_l)_{l= 1}^n$ be a uniformly distributed partition of $n$. Let $I=\sum_{l= 1}^n a_l$ be the number of parts in $(a_l)_{l\ge 1}$, then for all $x \in \reals$
    \[
        \lim_{n \to \infty}\Prob\left(\frac{\pi}{\sqrt{6n}}I - \log \frac{\sqrt{6n}}{\pi} \le x\right) = e^{-e^{-x}}.
    \]
\end{theorem}
Stated probabilistically, Theorem~\ref{thrm:fristedt1} states that $a_l$ converges to an exponential random variable when appropriately normalized. Likewise, Theorem~\ref{thrm:fristedt2} states that $I$ converges to a Gumbel distribution when appropriately normalized. Theorem~\ref{thrm:fristedt1} is actually a weaker statement than Theorem~2.2 in \citet{fristedt1993structure} which gives the limiting distribution of several entries $(a_1,\ldots,a_{l_n})$ for $l_n$ growing at a rate slower than $n^{1/4}$. Also, Theorem~2.3 in \citet{fristedt1993structure} is stated in terms of the largest part of $(a_l)_{l\ge 1}$. However, as noted in the introduction of \citet{fristedt1993structure}, the largest part has the same distribution as the number of parts. This can be seen be reflecting the Young diagram of the partition as in Figure~\ref{fig:young_diagrams}. Figure~\ref{fig:limit_laws} contains an illustration of Theorems~\ref{thrm:fristedt1} and \ref{thrm:fristedt2}, and our new algorithm. The limit theory is in accord with finite $n$ `reality.'

\begin{figure*}
    \centering
    \begin{subfigure}{.48\textwidth}
      \centering
      \includegraphics[width=\linewidth]{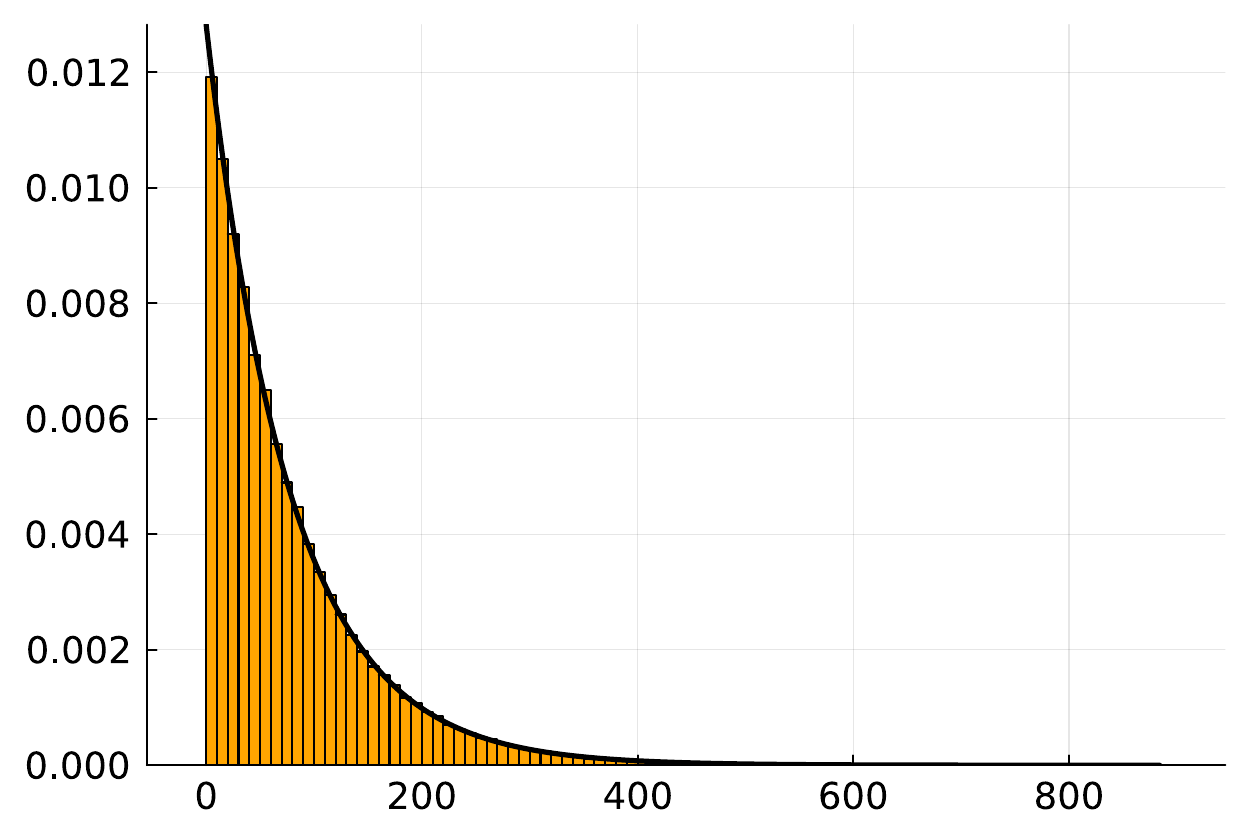}
      \caption{Number of ones ($n=10^4$)}
      \label{fig:sfig1}
    \end{subfigure}%
    \begin{subfigure}{.48\textwidth}
      \centering
      \includegraphics[width=\linewidth]{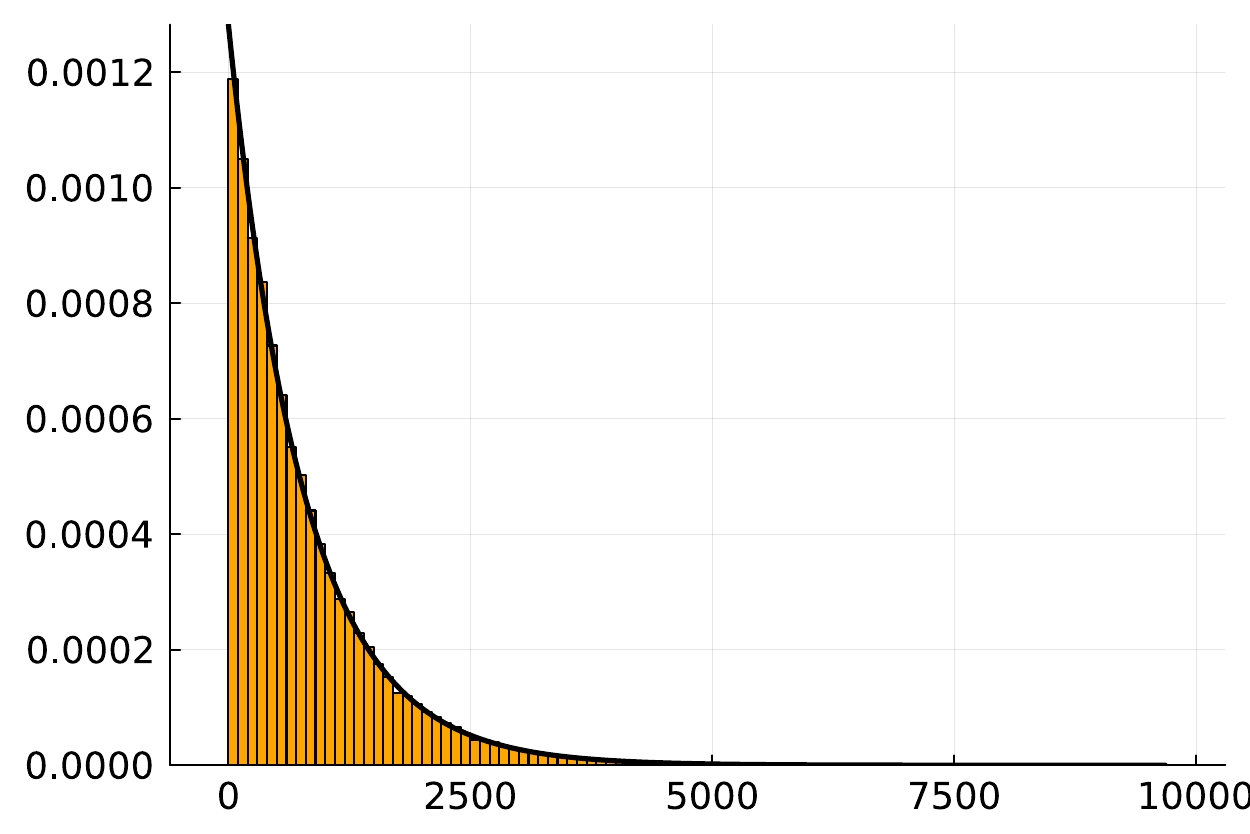}
      \caption{Number of ones ($n=10^6$)}
      \label{fig:sfig2}
    \end{subfigure}
    \begin{subfigure}{.48\textwidth}
        \centering
        \includegraphics[width=\linewidth]{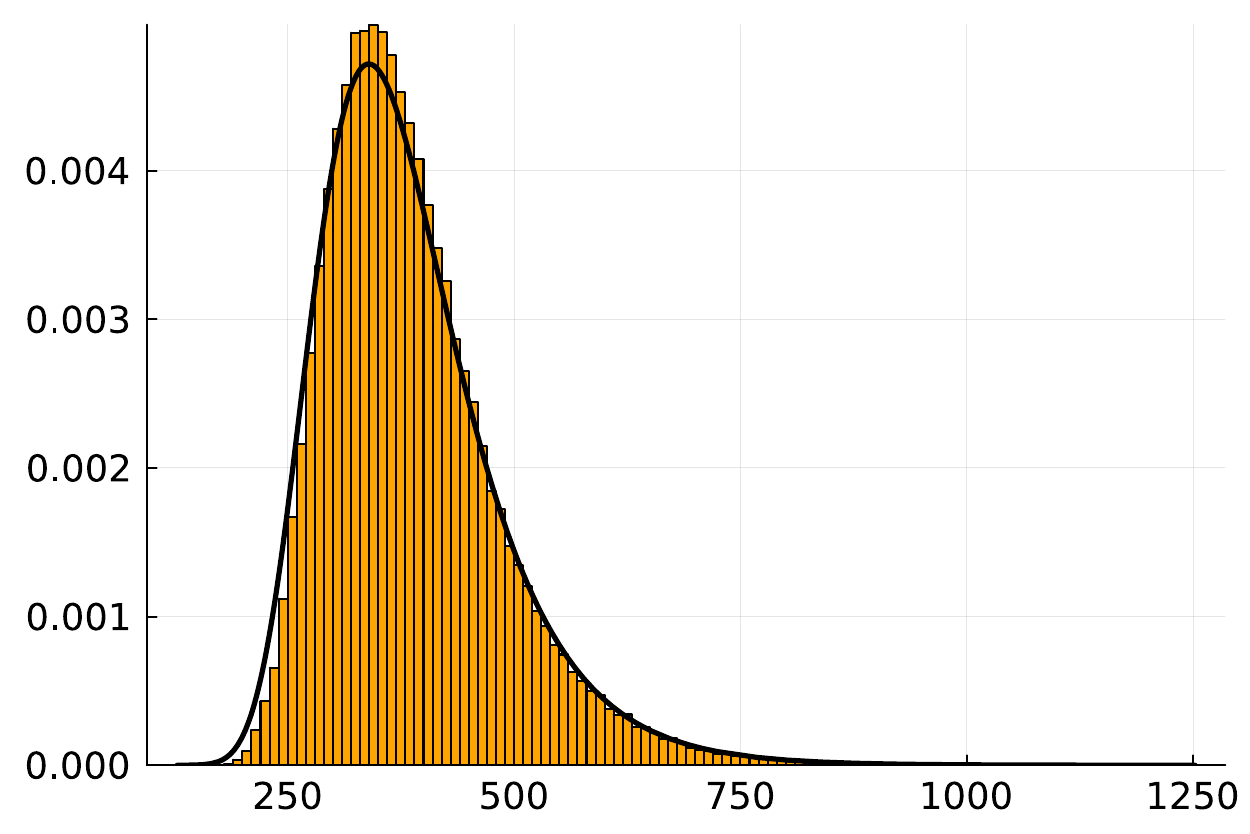}
        \caption{Number of parts ($n=10^4$)}
        \label{fig:sfig3}
      \end{subfigure}%
      \begin{subfigure}{.48\textwidth}
        \centering
        \includegraphics[width=\linewidth]{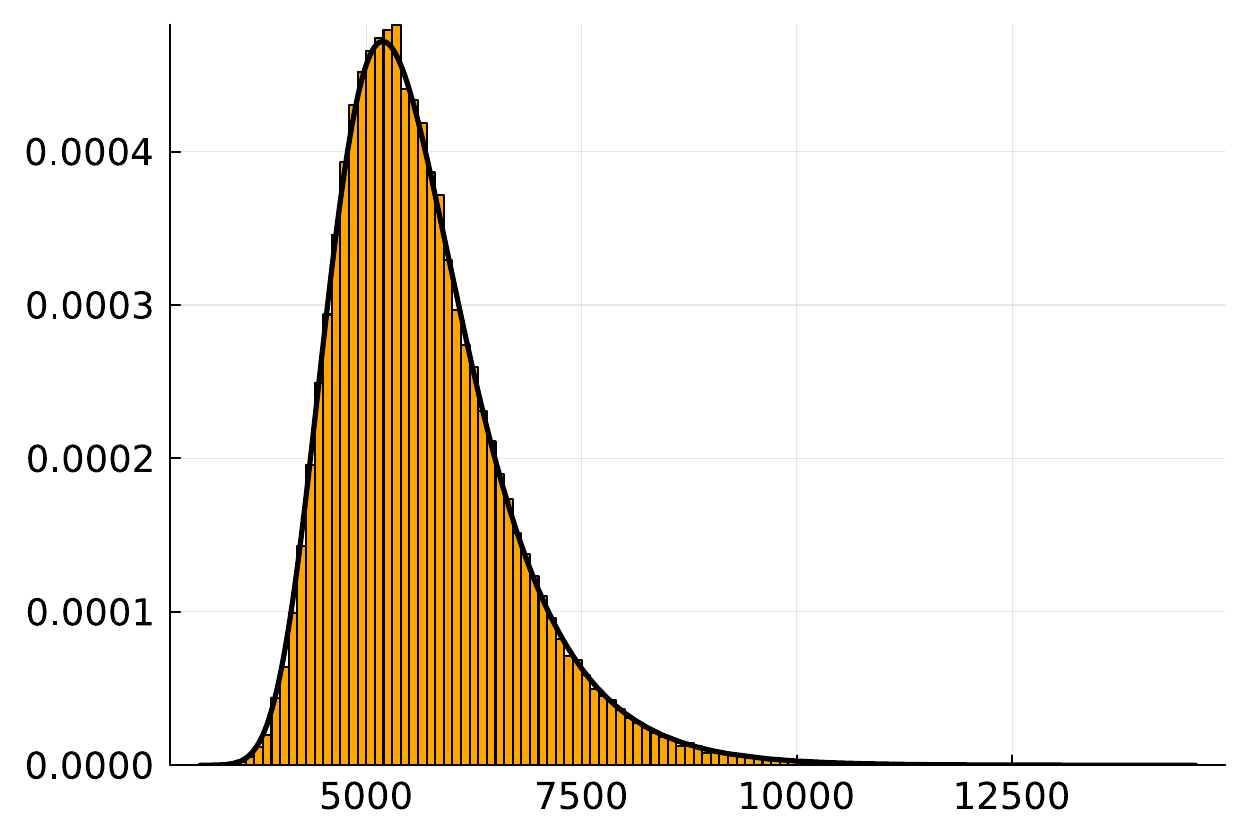}
        \caption{Number of parts ($n=10^6$)}
        \label{fig:sfig4}
      \end{subfigure}
    \caption{Comparisons between samples from the reflected Burnside process and the asymptotic distributions in Theorems~\ref{thrm:fristedt1} and \ref{thrm:fristedt2}. The histogram represents $10^5$ partitions generated from the reflected Burnside process (Section~\ref{sec:transposing}). Each partition was generated by taking $20$ steps of the reflected Burnside process initialized at the partition $1^n$. The dark lines show the asymptotic densities of the two features. Figures~\ref{fig:sfig1} and \ref{fig:sfig2} show that the number of ones is well approximated by an exponential distribution for both $n=10^4$ and $n=10^6$. Figure~\ref{fig:sfig3} shows some disagreement between the number of parts and the limiting Gumbel distribution and the fit is much better for $n=10^6$. Figures~\ref{fig:sfig2} and \ref{fig:sfig4} also suggest that the reflected Burnside process reaches its stationary distribution in just 20 steps for large $n$. On a personal computer, it took 5 minutes to sample all $10^5$ partition when $n=10^4$. For $n=10^6$, the same computer took just under an hour.}
    \label{fig:limit_laws}
\end{figure*}

\subsection{The Burnside process for partitions}\label{sec:lumped-parts}

This section describes our implementation of the lumped Burnside process on $\mathcal{P}_n$. We will see that lumping reduces the complexity of the Burnside process from order $n$ to order $\sqrt{n}(\log n)^2$. We begin by discussing the connection between partitions and conjugacy classes in permutation groups.

The permutation group, $S_n$, acts on itself by conjugation. Specifically, let $\X = G=S_n$ and consider the group action
\begin{equation}\label{eq:conjugate}
    \sigma^\tau = \tau^{-1} \sigma \tau.
\end{equation}
Under this action, two permutations are in the same orbit precisely when they are conjugate in $S_n$. Furthermore, two permutations are conjugate if and only if they have the same cycle type. That is, if $a_l$ is the number of $l$ cycles in $\sigma$, then the map 
\[
    \sigma \mapsto a = (a_l)_{l=1}^n,
\]
is a bijection between the orbits of the action \eqref{eq:conjugate} and partitions. For a proof of this fact, see \citet[Page~292]{suzuki1986group}. The partition $a$ is called the \emph{cycle type} of $\sigma$.

When applied to this group action, the Burnside process can be used to uniformly sample partitions by mapping a permutation to its cycle type. The two steps of the Burnside process are actually the same. From $\sigma$, uniformly sample $\tau$ such that $\sigma \tau=\tau\sigma$, then sample $\sigma'$ such that $\sigma'\tau=\tau\sigma'$. Because of this, we will slightly abuse terminology and in this case say that the Burnside process is
\begin{enumerate}
    \item From $\sigma \in S_n$, uniformly sample $\tau \in C_{S_n}(\sigma)$.
\end{enumerate}
The set $C_{S_n}(\sigma)=\{\tau \in S_n : \sigma\tau=\tau\sigma\}$ is called the \emph{centralizer} of $\sigma$. The structure of $C_{S_n}(\sigma)$ is well understood, and we will briefly review it here. For proofs and more details see, for example, page 295 of \citet{suzuki1986group}. 

The centralizer $C_{S_n}(\sigma)$ is isomorphic to a product of \emph{wreath products} which we will now define. For $l,a \ge 1$ let $C_l$ be the cyclic group of order $l$ and let $S_{a}$ be the symmetric group on $a$ elements. The wreath product $C_l^a \rtimes S_{a}$ is a subgroup of $S_{la}$. To understand the wreath product, think of $[la]$ as $a$ blocks of size $l$ ($\{1,\ldots,l\}$,$\{l+1,\ldots,2l\}$, and so on). Elements in $C_l^a \rtimes S_a$ can apply separate cyclic shifts to each block and permute blocks amongst themselves. For a more explicit description of the elements of $C_l^a \rtimes S_a$, see the start of the proof in Appendix~\ref{appn:lumped-parts}.

Wreath products relate to centralizers in the following way. A permutation $\tau$ is in $C_{S_n}(\sigma)$ if and only if for each $l$, $\tau$ cyclically shifts the $l$ cycles of $\sigma$ and permutes the $l$ cycles amongst themselves. Thus, if the cycle type of $\sigma$ is $(a_l)_{l=1}^n$, then 
\begin{equation}\label{eq:centralizer}
    C_{S_n}(\sigma) \cong \prod_{l : a_l \neq 0} C_l^{a_l} \rtimes S_{a_l},
\end{equation}
Equation~\eqref{eq:centralizer} shows that a permutation $\tau \in C_{S_n}(\sigma)$ can be represented as a list of disjoint permutations $(\tau^{(l)})_{l : a_l\neq 0}$ with $\tau^{(l)} \in C_l^{a_l} \rtimes S_{a_l}$. If $\mathcal{C}(\tau)$ is the set of cycles in $\tau$, then 
\[
    \mathcal{C}(\tau) = \bigsqcup_{l : a_l\neq 0} \mathcal{C}(\tau_l).
\]
That is, the cycles of $\tau$ are the union of the cycles of each $\tau^{(l)}$. Thus, if $b=(b_i)_{i=1}^n$ is the cycle type of $\tau$ and $b = (b^{(l)}_i)_{i=1}^n$ is the cycle type of $\tau_l$, then
\[
    b = \sum_{l:a_l \neq 0} b^{(l)}.
\]
Finally, the product in \eqref{eq:centralizer} implies that, under the uniform distribution on $C_{S_n}(\sigma)$, the permutations $(\tau_l)_{l : a_l \neq 0}$ are independent. We have thus arrived at the following lemma.
\begin{lemma}\label{lem:lumped-parts1}
    Fix a permutation $\sigma$ with cycle type $(a_l)_{l = 1}^n$. Let $\tau$ be a uniform sample from $C_{S_n}(\sigma)$ and let $b = (b_i)_{i = 1}^n$ be the cycle type of $\tau$. Then
    \[ 
        b \stackrel{d}{=} \sum_{l : a_l \neq 0} b^{(l)},
    \] 
    where $b^{(l)} = (b^{(l)}_i)_{i=1}^n$ is the cycle type of a uniformly drawn element of $C_l^{a_l} \rtimes S_{a_l}$.
\end{lemma}
Lemma~\ref{lem:lumped-parts1} implies that to run the lumped Burnside process on $\mathcal{P}_n$ it is sufficient to know how to sample the cycle type of permutation $\tau$ uniformly sampled from $C_l^{a} \rtimes S_{a} \subseteq S_{la}$. \citet{diaconis2024poisson,tung2025cutting} studied the cycle type of uniformly sampled elements of $\Gamma^a \rtimes S_a$ for general $\Gamma \subseteq S_l$, It is easier to state their results in terms of \emph{cycle lengths} instead of cycle types. 
\begin{definition}
    Let $\tau$ be a permutation with disjoint cycles $(B_j)_{j=1}^m$. The \emph{cycle lengths} of $\tau$ is the vector $(\lambda_j)_{j=1}^m$ where $\lambda_j = |B_j|$.
\end{definition}
Note that if $(\lambda_j)_{j=1}^m$ are the cycle lengths of $\tau$, then the cycle type of $\tau$ is given by 
\[
    b_i =\sum_{j=1}^m I[\lambda_j=i],
\]
where $I[A]=1$ if $A$ is true and $I[A]=0$ otherwise. 

It is shown in \citet{diaconis2024poisson} that the cycle lengths of a uniformly sampled element in $C_l^a \rtimes S_a$ can be constructed by first sampling the cycle lengths $(\lambda_j)_{j=1}^m$ of a uniform permutation of length $a$. Then, for each $j$, $\lambda_j$ is replaced with $d_j$ copies of $\lambda_j k_j$ where $d_j$ is a divisor of $l$ and $k_j = l/d_j$. The divisors $d_j$ are chosen from a specific distribution as stated more formally below. 
\begin{proposition}\label{prop:lumped-parts}
    Fix $l,a \ge 1$ and let $(\lambda_j)_{j=1}^m$ be the cycle lengths of a uniformly distributed element of $S_{a}$. Let  $(U_j)_{j=1}^m$ be independent and uniformly sampled from $[l]$. Set $d_j = \gcd(l,U_j)$ and $k_j=l/d_j$. 

    If $\mu_{j,p}=\lambda_k k_j$ for all $j \in [m]$ and $p \in [d_j]$, then the list of all $\mu_{j,p}$ is equal in distribution to the cycle lengths of a uniformly sampled element of $C_l^{a}\rtimes S_{a}$.

    Equivalently, if $b=(b_i)_{i=1}^{la}$ is the cycle type of a uniformly distributed element of $C_l^a \rtimes S_a$, then
    \begin{equation*}
        b_i \stackrel{d}{=} \sum_{j=1}^m d_j I[\lambda_j k_j = i],
    \end{equation*} 
    where $I[A]=1$ if $A$ is true and $I[A]=0$ otherwise.
\end{proposition}
For completeness, a proof of Proposition~\ref{prop:lumped-parts} is given in Appendix~\ref{appn:lumped-parts}. The computational benefit of Proposition~\ref{prop:lumped-parts} is that the cycle lengths $(\lambda_j)_{j=1}^m$ can be sampled directly without having to sample a permutation. In our implementation we sample $(\lambda_j)_{j=1}^m$ via discrete stick breaking.
\begin{definition}\label{def:stick-breaking}
    The \emph{discrete stick breaking distribution} on $n$ is a distribution over sequences $(\lambda_j)_{j=1}^m$ with $\lambda_j \in \integers_{\ge 0}$ and $\sum_{j=1}^m \lambda_j =  n$. To sample from the discrete breaking distribution first sample $\lambda_1$ uniformly from $[n]$. If $\lambda_1 = 0$, then stop. Otherwise, sample $\lambda_2$ uniformly from $[n-\lambda_1]$. If $\lambda_1+\lambda_2=n$, then stop and otherwise sample $\lambda_3$ uniformly from $[n-\lambda_1-\lambda_2]$. Continue in this way to produce $(\lambda_j)_{j=1}^m$ with $\sum_{j=1}^m \lambda_j=n$.  
\end{definition}
It is well--known that the discrete stick breaking distribution is the same as the distribution of the cycle lengths of a randomly drawn permutation $\sigma \in S_n$. The discrete stick breaking distribution is also an instance of Ewens distribution \citep{crane2016ubiquitous}.

Together Lemma~\ref{lem:lumped-parts1}, Proposition~\ref{prop:lumped-parts} and Definition~\ref{def:stick-breaking} give an efficient implementation of the lumped Burnside process. If we are currently at the partition $a=(a_l)_{l \ge 1}$, then a single step is as follows
\begin{enumerate}
    \item For each $l$ with $a_l \neq 0$ independently sample $b^{(l)}=(b_{i}^{(l)})_{i \ge 1}$ using Proposition~\ref{prop:lumped-parts}.
    \item Return $b = \sum_{l: a_l \neq 0} b^{(l)}$.
\end{enumerate}

A more detailed description is given in Algorithm~\ref{alg:partitions}. The following examples also give a feel for the procedure.

\begin{algorithm*}[hbt]
    \caption{Lumped Burnside process for partitions}
    \begin{algorithmic}[1]
    \Require{$(a_l)_{l=1}^n$ (partition of $n$)}
    \Ensure{$(b_i)_{i=1}^n$ (new partition sampled from the Burnside process)}
    \State{$b_i \gets 0$ for all $i \in [n]$}
    \For{$l\gets 1$ to $n$}
        \If{$a_l \neq 0$}
            \State{Sample $(\lambda_j)_{j=1}^m$ from the discrete stick breaking distribution on $a_l$}
            \For{$j \gets 1$ to $m$}
                \State{Sample $U$ uniformly from $[l]$}
                \State{$d \gets \gcd(U,l)$}
                \State{$i \gets \lambda_j l/d$}
                \State{$b_i \gets b_i + d$}
            \EndFor
        \EndIf
    \EndFor
    \State \Return{$(b_i)_{i=1}^n$}
\end{algorithmic}
\label{alg:partitions}
\end{algorithm*}

\begin{example}
    Suppose $a = (n,0,0,\ldots)$ so that $a_l \neq 0$ only for $l=1$. In this case, $(b_i)_{i = 1}^n$ is simply the cycle type of a uniformly sampled permutation in $S_n$. Equivalently, $b_i=\#\{j:\lambda_j=i\}$ where $(\lambda_j)_{j=1}^m$ is drawn from the discrete stick breaking distribution on $n$. 
\end{example}
\begin{example}\label{ex:cyclic}
    Suppose $a = (0,0,\ldots,0,1)$ so that $a_n=1$ and $a_l=0$ otherwise. The new partition $(b_i)_{i = 1}^n$ has exactly one non-zero entry $b_{n/d} = d$ where $d$ is a divisor of $n$. The probability of a particular value of $d$ is $\phi(n/d)/n$ where $\phi$ is Euler's totient function. In other words, the partition $(b_i)_{i \ge 1}$ is the cycle type of a uniformly chosen element of $C_n \subseteq S_n$.
\end{example}
The main benefit of Algorithm~\ref{alg:partitions} is that we only have to store and use the non-zero values of $(a_l)_{l = 1}^n$. This can dramatically reduce the memory and time requirements of the algorithm. Indeed, the constraint $\sum_{1\ge1} la_l = n$ implies that at most $2\sqrt{n}$ values of $a_l$ can be non-zero and that each $a_l$ is at most $n$. Thus, the memory required to store $(a_l)_{l \ge 1}$ is $O(\sqrt{n}\log n)$.

The average time complexity of the lumped Burnside process is at most $O(\sqrt{n}(\log n)^2)$. Sampling $(\lambda_j)_{j=1}^m$ from the discrete stick breaking distribution on $a_l$ has average complexity $(\log a_l)^2$. The complexity of sampling uniformly from $[l]$ and computing the greatest common divisor $d$ has complexity $O(\log l)$ and this has to be done $O(\log a_l)$ times on average. Thus, the average case complexity of the algorithm is
\begin{align*}
    &O\Bigg(\sum_{l:a_l \neq 0} (\log a_l)^2+\log l\log a_l\Bigg)\\
    &=O\Bigg((\log n)^2 |\{l : a_l \neq 0\}| \Bigg)\\
     &=O\big(\sqrt{n}(\log n)^2\big).
\end{align*}
In contrast, the original Burnside process on $S_n$ requires storing a permutation of size $n$ and so memory requirements on the order of $n \log n$. This is because a permutation of size $n$ contain $n$ numbers most of which are on the order of $n$. Likewise, computing the cycle type of a permutation given as a list of numbers also has linear complexity. Thus,  a single step of the original process has time and memory complexity of order $n \gg \sqrt{n}(\log n)^2$. These theoretical arguments are supported by our experiments which are reported in Figure~\ref{fig:partition-run-time}. 

However, while the per-step complexity of the lumped Burnside process is very low, we have found that it does not always mix quickly. Indeed, consider the process started a partition with a single part (meaning $a_n=1$). Example~\ref{ex:cyclic} shows that the chance that the chain remains at $a_n=1$ is $\phi(n)/n$. If $n$ is prime, this chance of holding is $1-1/n$ and the mixing time is bounded below by $n$. Similarly, Figure~\ref{fig:largest-part} shows the Burnside process is initialized at the all 1's partition (meaning $a_1=n$). In this simulation, the largest part is constant for long periods of time providing a barrier to mixing. The next section introduces a second Markov chain on $\mathcal{P}_n$ that addresses this concern.

\begin{figure*}
    \centering
    \includegraphics[width = 0.8\textwidth]{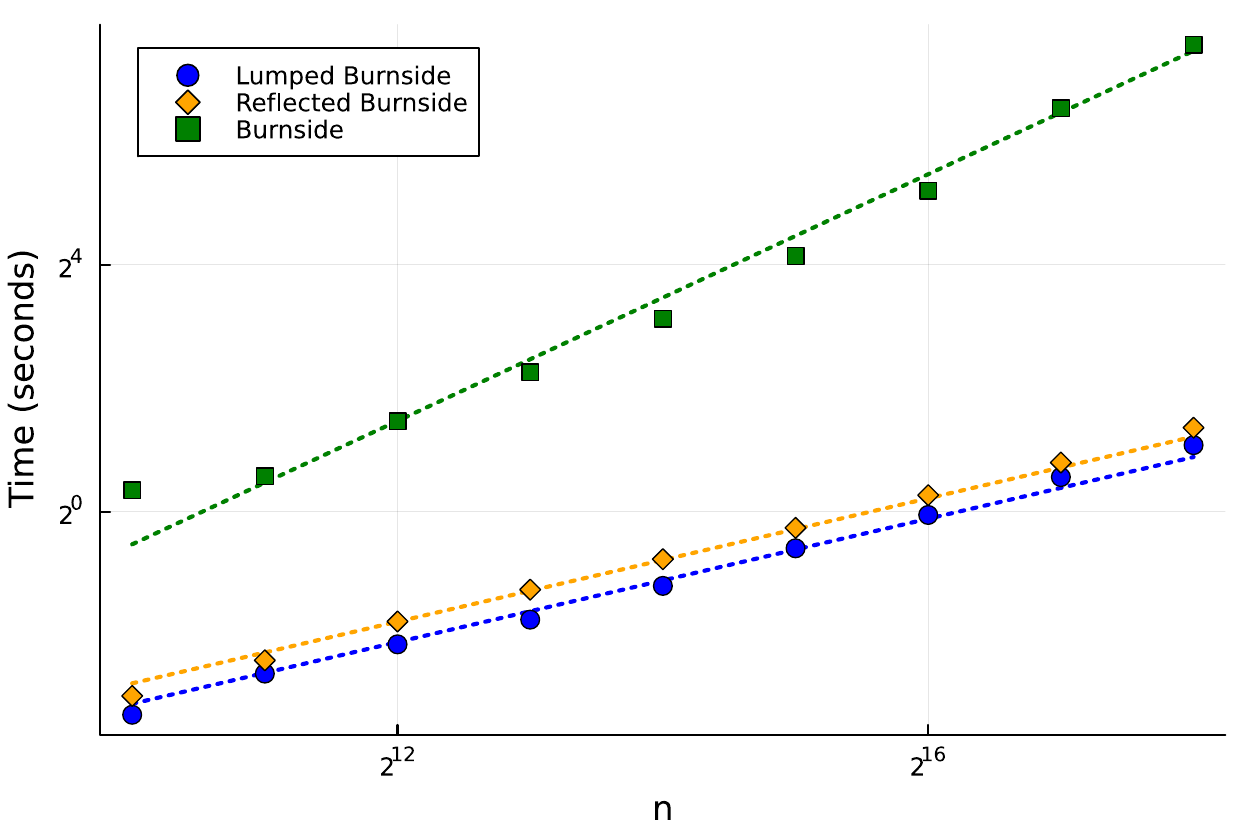}
    \caption{The above is a log--log plot showing the benefits of the lumped process for partitions. Each point corresponds to running the corresponding Burnside process for 10,000 steps. The dashed lines have fixed slopes (1 for the Burnside process, 1/2 for the lumped process and reflected process) and fitted intercepts. The good fit these lines supports our analysis that, up to log factors, the original process has complexity of order $n$ and the lumped process has complexity of order $\sqrt{n}$. Note also the comparable per-step complexity of the lumped and reflected processes.}
    \label{fig:partition-run-time}
\end{figure*}

\subsection{Speeding up the Burnside process}\label{sec:transposing}

The \emph{reflected Burnside process} is a second Markov chain on $\mathcal{P}_n$ that empirically mixes much faster than the Burnside process. Importantly, the reflected Burnside process has roughly the same per-step complexity as the lumped Burnside process from Section~\ref{sec:lumped-parts}.  Recall that the transpose of a partition is defined by reflecting the Young diagram of the partition as shown in Figure~\ref{fig:young_diagrams}. 

\begin{figure*}[tbp]
    \centering
    \begin{tikzpicture}
        \foreach \y/\num in {0/4, 1/3, 2/1} {
            \foreach \x in {0,1,2,3,4} {
                \ifnum\x<\num
                    \draw[thick] (\x,-\y) rectangle (\x+1,-\y-1);
                \fi
            }
        }

        \foreach \y/\num in {0/3, 1/2, 2/2, 3/1} {
            \foreach \x in {0,1,2,3} {
                \ifnum\x<\num
                    \draw[thick] (\x+7,-\y) rectangle (\x+8,-\y-1); 
                \fi
            }
        }
    \end{tikzpicture}
    \caption{Young diagrams for two partition of size  $a=(1,0,1,1,0,\ldots,0)\in \mathcal{P}_8$ on the left and its transpose $b=(1,2,1,0,\ldots,0)$ on the right.}
    \label{fig:young_diagrams}
\end{figure*}

Let $P$ be the transition kernel for the lumped Burnside process on $\mathcal{P}_n$ and let $\Pi : \mathcal{P}_n\times \mathcal{P}_n \to \{0,1\}$ be the matrix with 
\[
\Pi(a,b) = \begin{cases}
    1 & \text{if } b \text{ is the transpose of } a,\\
    0 & \text{otherwise}.
\end{cases} 
\]
The reflected Burnside process is the Markov chain with transition kernel $Q=\Pi P$. That is, to take one step of the reflected Burnside process, first compute the transpose of the current partition and then take a step of the lumped Burnside process via Algorithm~\ref{alg:partitions}. Since transposing is a bijection from $\mathcal{P}_n$ to $\mathcal{P}_n$, the transition kernel $Q$ is still ergodic with uniform stationary distribution. However, our simulations show that the transition matrix $Q$ mixes much more quickly than the original $P$. For example, Figure~\ref{fig:alternating-parts} shows that, under $Q$, the number of parts and the size of the largest part quickly approach the value predicted by Theorem~\ref{thrm:fristedt2}. Under $P$, both features are far from their stationary distribution even after many steps.

The findings in Figure~\ref{fig:alternating-parts} are consistent with a number of other simulations for different features and different values of $n$. For every feature that we have considered, the mixing of time of $Q$ seems to grow logarithmically in $n$ and remains less than $50$ for $n$ as large as $10^{10}$.

An algorithm for computing the transpose of a partition in exponential notation is given in Algorithm~\ref{alg:transpose} in the appendix. The complexity of Algorithm~\ref{alg:transpose} is on the order of $\sqrt{n}\log n$. This means that using $Q$ instead of $P$ increases the per-step complexity by at most a constant factor. The mixing time benefits of $Q$ far outweigh the moderate increase in per-step complexity.

The reflected Burnside process is an example of speeding up a Markov chain by adding a deterministic jump that preserves the stationary distribution. This technique is well--established in the literature \citep{chung1987random,chatterjee2020speeding,hermon2022universality,ben2023cutoff}. We do not know why taking the transposition at each step leads to such a speed-up. Our intuition is that the Burnside process mixes well when started at partition with many parts and mixes poorly when started at a partition with large parts. Taking the transpose thus interchanges a region where the chain mixes slowly with a region where the chain mixes rapidly while preserving the uniform stationary distribution. Turning this intuition into a quantitative statement and connecting this example to the above literature are open problems. 

\begin{figure*}[t]
    \centering
    \begin{subfigure}{0.48\textwidth}
        \centering
        \includegraphics[width=\textwidth]{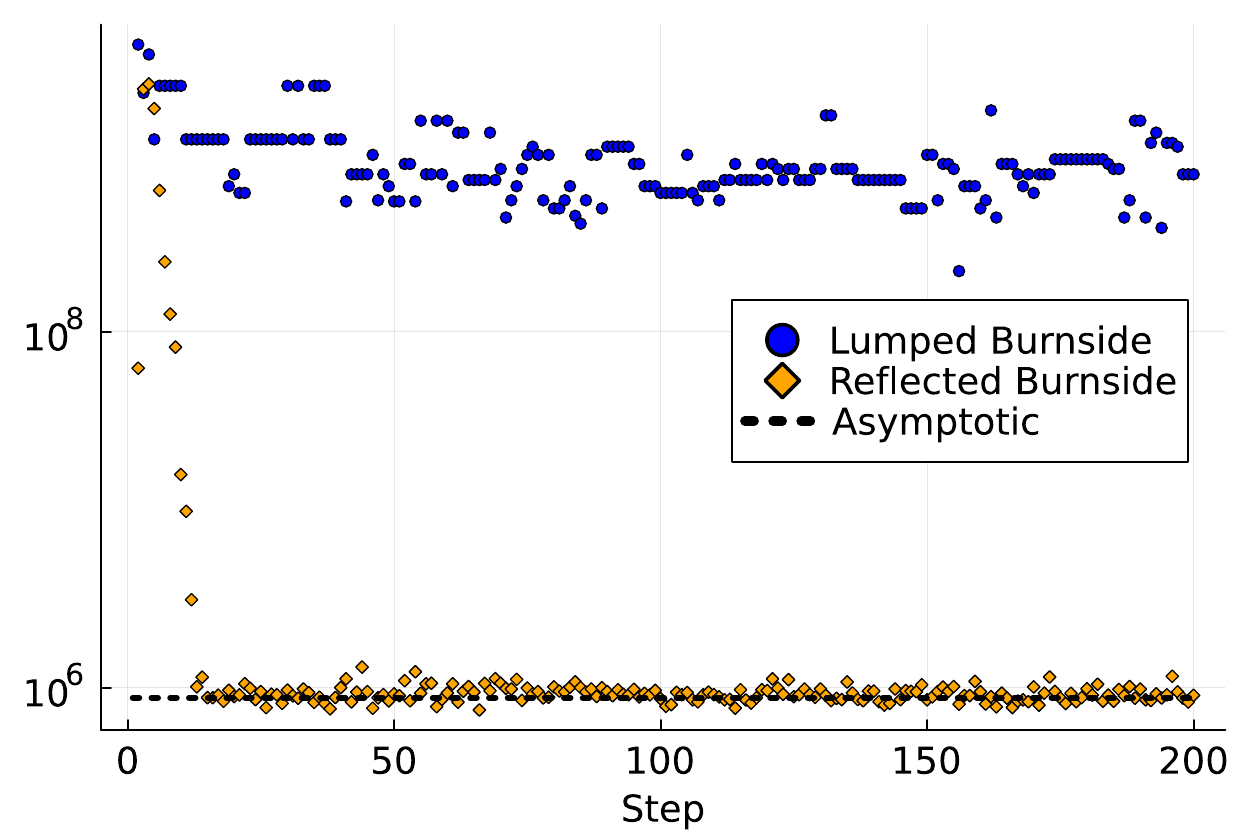}
        \caption{Largest part under the Burnside process $(n=10^{10})$.}
        \label{fig:largest-part}
    \end{subfigure}
    \hfill
    \begin{subfigure}{0.48\textwidth}
        \centering
        \includegraphics[width=\textwidth]{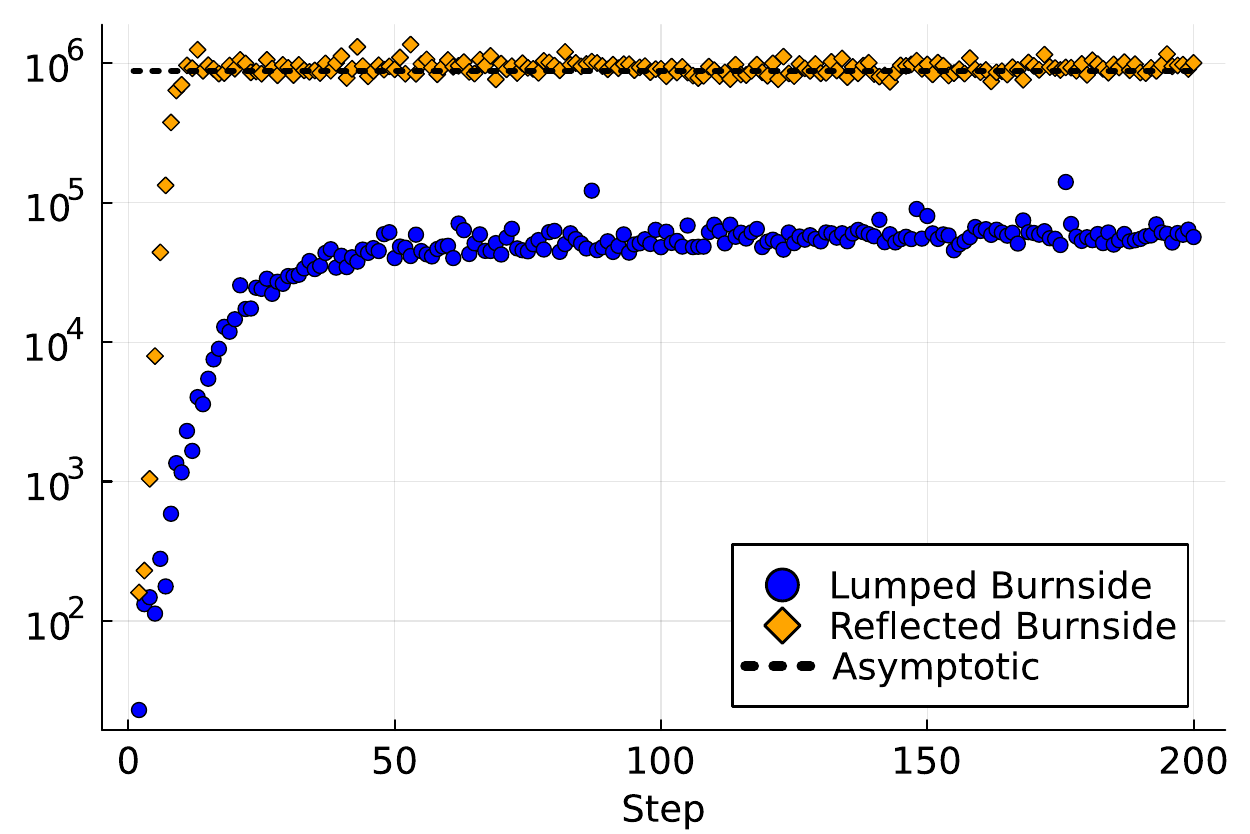}
        \caption{Number of parts under the Burnside process $(n=10^{10})$.}
        \label{fig:number-of-parts}
    \end{subfigure}
    \caption{The above figures compares the reflected Burnside process to lumped Burnside process for generating partition of size $n=10^{10}$. In each figure, both processes are initialized at the partition $1^n$ and run for 200 steps. Figure~\ref{fig:largest-part} shows the evolution of the size of the largest part and Figure~\ref{fig:number-of-parts} shows the evolution of the number of parts. Theorem~\ref{thrm:fristedt2} from \citep{fristedt1993structure} gives the asymptotic distribution of both these statistics under the uniform distribution and the asymptotic mean is represented as a dashed line. Under the reflected Burnside process, both statistics quickly converge to the uniform distribution. In contrast, neither statistic has converged under the lumped Burnside process. The lumped Burnside process also exhibits strong autocorrelation in Figure~\ref{fig:largest-part}. This is discussed in Section~\ref{sec:transposing}.}
    \label{fig:alternating-parts}
\end{figure*}

\section{Conclusion}

Classifying the orbits of a finite group acting on a finite set is a huge, unmanageable topic (!). The survey \citet{Keller2003Orbits} points to the richness and depth within group theory. Our examples show that there are worthwhile applications further afield in probability and statistics. 

The analysis and simulations in Sections~\ref{sec:tables} and \ref{sec:parts} demonstrate the computational benefits of the lumped Burnside process. Even for partitions or contingency tables of moderate size, the lumped chain is several orders of magnitude faster than original process. We have found two common ideas that help achieve the full benefits of lumping. These are using simple data structures for the orbits, and using discrete stick breaking to generate the cycle type of a random permutation. We believe that both of these will help in other examples, and we give some more details below.

In both Sections~\ref{sec:tables} and \ref{sec:parts} the orbits of the group action could be easily identified. In Section~\ref{sec:tables} the orbits are in bijection with contingency tables and in Section~\ref{sec:parts} the orbits are in bijection with partitions. These bijections meant we could run the Burnside process directly on the orbits. This led to an implementation that required far less memory than the original Burnside process. We also found that the choice of how to represent the orbits matters. In Section~\ref{sec:parts} we represented the partition in exponential notation and found that this worked better than the representing the partition as a list of parts. Unfortunately, in many other examples there is not a known identification of the orbits. Thus, this aspect of lumping may not be implementable in all other examples.

The second aspect of lumping that we found helpful is discrete stick breaking. Lemmas~\ref{lem:lumped_tables2} and \ref{lem:lumped-parts1} describe the distribution of the cycle type of a uniformly distributed element of a stabilizer of the group action. These results are useful because in both cases the second step of the Burnside process only depends on these cycles and not on the full permutation. We believe that this should hold in some generality. Specifically, when the group $G$ is a subgroup of a permutation group, then some form of stick breaking could be used to generate the cycle type of a uniformly chosen element of a stabilizer. This technique could still be applied in cases when the orbits do not have a known representation. 

For partitions, we also found that reflected Burnside process in Section~\ref{sec:transposing} converges much more quickly than the original Burnside process. The transpose map used by the reflected Burnside process is a bijection on partitions. However, the transpose map does not lift to a bijection on permutations. Thus, another benefit of lumping is that it can lead to new Markov chains that combine the Burnside process with deterministic bijections.

Finally, we end with two comments about implementation. First, the choice of programming language matters. We originally implemented our algorithms in Python but later moved to Julia. In Julia, the programs ran 100 times faster and could scale to larger problems. Second, both of our algorithms naturally lend themselves to parallelization. We did not explore this in our implementation, but it could lead to further speed--ups.

\bmhead{Acknowledgements}

We thank Timothy Sudijono, Nathan Tung, Chenyang Zhong, Andrew Lin, Bal\'azs Gerencs\'er and Laurent Bartholdi for helpful discussions and comments. We also thank the referees for detailed, constructive suggestions. This research was partially funded by NSF grant 1954042.

\begin{appendices}

\section{Proofs}

\subsection{Proof of Proposition~\ref{prop:lump}}\label{appn:lump}

The following lemma states that if two elements are in the same orbit, then their stabilizers are conjugate.
\begin{lemma}\label{lem:conjugate-stabilizers}
    For all $u,x \in \X$ and $s \in G$, if $x^s = u$, then 
    \[G_u = s^{-1}G_x s \]
    Furthermore, if $g = s^{-1}hs$, then $\X^g =  \{x^s : x \in \X_h\}$
\end{lemma}
\begin{proof}
    Suppose that $x^s = u$. Then for all $g \in G_x$,
    \[
        u^{s^{-1}gs} = (x^{g})^s = u^s = x.
    \]
    And conversely, if $g \in G_u$, then $u^{sgs^{-1}}=u$ and hence $G_u = s^{-1}G_xs$.

    For the second claim, suppose that $g=s^{-1}hs$, then for all $x \in \X_h$
    \[
     (x^s)^{g} = x^{hs} = (x^h)^s = x^s.
    \]
    And conversely, if $x \in \X^g$, the $x^{s^{-1}} \in \X_h$ and so $\X^g = \{x^s : x \in \X_h\}$.
\end{proof}
We will now prove Proposition~\ref{prop:lump}.
\begin{proof}
    Recall that 
    \begin{align*}
        P(x,y) &= \frac{1}{|G_x|}\sum_{g \in G_x \cap G_y}\frac{1}{|\X^g|},
    \end{align*} 
    and
    \[
    \overline{P}(\O_x,\O_y) = \sum_{z \in \O_y}P(x,z).
    \]
    The first claim in Proposition~\ref{prop:lump} is that the definition of $\overline{P}$ does not depend on the choice of $x \in \O_x$. Thus, suppose that $u \in \O_x$. Then there exists $s \in G$ such that $G_u = s^{-1}G_xs$ by Lemma~\ref{lem:conjugate-stabilizers}. Thus,
    \begin{align*}
        &\sum_{z \in \O_y}P(u,z)\\
        &=\sum_{z \in \O_y}\frac{1}{|G_u|}\sum_{g \in G_u \cap G_z}\frac{1}{|\X^g|}\\
        &=\sum_{z \in \O_y}\frac{1}{|s^{-1}G_xs^{-1}|}\sum_{g \in s^{-1}G_xs \cap G_z}\frac{1}{|\X^g|}\\
        &=\sum_{z \in \O_y}\frac{1}{|G_x|}\sum_{g \in G_x \cap sG_zs^{-1}}\frac{1}{|\X^g|}\\
        &=\sum_{z^s \in \O_y}\frac{1}{|G_x|}\sum_{g \in G_x \cap sG_{z^{s}}s^{-1}}\frac{1}{|\X^g|}\\
        &=\sum_{z^s \in \O_y}\frac{1}{|G_x|}\sum_{g \in G_x \cap G_{z^{s}}}\frac{1}{|\X^g|}\\
        &=\sum_{z \in \O_y} P(x,z).
    \end{align*}
    Thus, the definition of $\overline{P}(\O_x,\O_y)$ does not depend on $x$. 
    
    The kernel $\overline{P}$ has positive entries and is therefore ergodic. Since $P$ is reversible with respect to $\pi(x) = \frac{1}{Z|\O_x|}$, it follows that $\overline{P}$ is also reversible with stationary measure
    \[\overline{\pi}(x) = \sum_{z \in \O_x} \pi(z) = \sum_{z \in \O_x}\frac{1}{Z|\O_z|} = \frac{1}{Z}. \]
    Thus, $\overline{P}$ is reversible with respect to the uniform distribution and hence symmetric. 

    Finally, since the definition of $\overline{P}(\O_x,\O_y)$ does not depend on $x \in \O_x$, Dynkin's criteria \citep[Section~6.3]{kemeny1976finite} implies that $\overline{P}$ is the transition kernel for the lumped Burnside process. 
\end{proof}

\subsection{Proof of Lemma~\ref{lem:cosets}}\label{appn:cosets}

\begin{proof}
    The first claim is simply stating that $h^{-1}sk=s$ if and only if $k=s^{-1}hs$. This implies that if $(h,k) \in (H \times K)_s$, then $k$ is determined by the relation $k=s^{-1}hs$. The set of possible $h$ is thus
    \[
        H \cap \{h : s^{-1}hs \in K\} = H \cap sKs^{-1}.
    \]
    This implies that $h \mapsto (h,s^{-1}hs)$ is a bijection between $H \cap sKs^{-1}$ and $(H \times K)_s$ which is the second claim in Lemma~\ref{lem:cosets}. For the final claim, suppose that there exists some $s \in G^{h,k}$. Then $hs=sk$ and furthermore $t \in G^{h,k}$ if and only if $ht = tk$. Thus, $t \in G^{h,k}$ is equivalent to $h(ts^{-1}) =(ts^{-1})h$. This implies that $G^{h,k}=C_G(h)s$.
\end{proof}

\subsection{Proof of Lemma~\ref{lem:lumped_tables}}\label{appn:lumped_tables}
\begin{proof}
    Let $\tau$ be drawn uniformly from $S_n^{h,k}\neq \emptyset$. Then, by Lemma~\ref{lem:cosets}, $\tau$ is drawn uniformly from the set $\{\tau \in S_n : k = \tau^{-1}h\tau\}$. In particular, the condition $S_n^{h,k} \neq \emptyset$ implies that $h=(h_i)_{i=1}^I\in S_\lambda$ and $k=(k_j)_{j=1}^I \in S_\mu$ are conjugate in $S_n$. Furthermore, for every $l$, the permutation $\tau^{-1}$ induces a bijection from the $l$-cycles of $k$ to the $l$-cycles of $h$. Since $\tau$ is uniformly distributed the bijection on the $l$-cycles is also uniformly distributed. The value $X_{i,j}^{(l)}$ is the number of $l$ cycles of $h$ in $L_i$ that $\tau^{-1}$ maps to an $l$ cycle of $k$ in $M_j$. Since the bijection on the $l$ cycles is uniformly distributed, $X^{(l)}$ has the Fisher--Yates distribution with row sums $r^{(l)}$ and column sums $c^{(l)}$.

    For the last claim, we have
    \begin{align*}
        &f(\tau)_{i,j} \\
        &=|L_i \cap \tau(M_j)|\\
        &=\sum_{l \ge 1} |\{x : x \in C, |C|=l, C \in \mathcal{C}(h_i) \cap \tau(\mathcal{C}(k_j)) \}|\\
        &=\sum_{l \ge 1} l\{C : |C|=l, C \in \mathcal{C}(h_i)\cap  \tau(\mathcal{C}(k_j))\}|\\
        &=\sum_{l \ge 1} lX_{i,j}^{(l)},
    \end{align*}
    as required.
\end{proof}

\subsection{Proof of Lemma~\ref{lem:lumped_tables2}}\label{appn:lumped_tables2}
\begin{proof}
    Let $(h,k)$ be drawn uniformly from $(S_\lambda \times S_\mu)_\sigma$ where $\sigma$ is such that $f(\sigma)=T$. By Lemma~\ref{lem:cosets}, we know that $h$ is uniformly drawn from $S_\lambda \cap \sigma S_\mu \sigma^{-1} = S_{L \land \sigma(M)}$. It follows that $h = (h_{i,j})_{i,j}$ where $h_{i,j}$ is a uniformly drawn permutation of $L_i \cap \sigma(M_j)$ and $(h_{i,j})_{i,j}$ are independent. The vector $(a_{i,j}^{(l)})_{l \ge 0}$ is simply the cycle type of $h_{i,j}$. Since $|L_i \cap \sigma(M_j)|=T_{i,j}$, the vector $(a_{i,j}^{(l)})_{l \ge 0}$ has the discrete stick breaking distribution on $T_{i,j}$. Independence of $(a_{i,j})_{i,j}$ follows from the independence of $(h_{i,j})_{i,j}$.

    Furthermore, since $h \in S_\lambda \cap \sigma S_\mu \sigma^{-1}$, every cycle of $h$ must be contained in $L_i \cap \sigma(M_j)$ for some $i$ and $j$. Thus,
    \begin{align*}
    &r_i^{(l)}\\
     &= |\{C \in \mathcal{C}(h_i): |C|=l, C \subseteq L_i\}|\\
    & = \sum_{j=1}^J |\{C \in \mathcal{C}(h_i) : |C|=l, C \subseteq L_i \cap \sigma(M_j)\}|\\
    & = \sum_{j=1}^J a_{i,j}^{(l)}.
    \end{align*}
    Finally, by Lemma~\ref{lem:cosets} $k=\sigma^{-1} h \sigma$. This implies that if $C$ is an $l$-cycle of $h$ in $L_i \cap \sigma(M_j)$, then $\sigma^{-1}(C)$ is an $l$-cycle of $k$ in $\sigma(L_i) \cap M_j$. And so, by similar reasoning, $c^{(l)}_j = \sum_{i=1}^I a_{i,j}^{(l)}$.
\end{proof}

\subsection{Proof of Proposition~\ref{prop:lumped-parts}}\label{appn:lumped-parts}
\begin{proof}
    Let $\tau = (V_1,V_2,\ldots,V_{a_l}; \pi)$ be a uniformly sampled element of $C_{l}^{a} \rtimes S_{a}$. Thus, $\pi$ is uniformly distributed in $S_{a}$, $V_1,\ldots,V_{a}$ are uniformly distributed on $C_l$ and $(V_1,\ldots,V_{a};\pi)$ are mutually independent.  As an element of $S_{la}$, the permutation $\tau$ cyclically shifts the elements of $L_1=\{1,\ldots,l\}$ by $V_1$, the elements of $L_2=\{l+1,\ldots,2l\}$ by $V_2$, \ldots, the elements of $L_{a} = \{(a-1)l+1,\ldots,al\}$ by $V_{a}$ and then permutes the blocks $\{L_1,\ldots,L_{a}\}$ according to $\pi$. Thus, if $i = kl+r$ with $r \in \{1,\ldots,l\}$, then $\tau(i) = \pi(k)l + (r+V_k \bmod l)$. Our goal is to understand the cycles of $\tau$. 

    Let $B_1,\ldots,B_m$ be the cycles of $\pi$ and let $\lambda_j=|B_j|$ be the length of the $j$th cycle and write $B_j = (k_1,k_2,\ldots,k_{\lambda_j})$. The permutation $\tau$ thus cycles the blocks $L_{k_1},L_{k_2},\ldots,L_{k_{\lambda_j}}$. When applying $\tau$ $\lambda_j$ times, an element $i \in L_{k_1}$ returns to the block $L_{k_1}$ but will be shifted by $U_j = V_{k_1}+V_{k_2}+\cdots+V_{k_{\lambda_j}}$. Thus, the size of the cycle of $\tau$ that contains $i$ is $\lambda_j$ times the order of $U_j$ in $C_l$. The order of $U_j$ is $l/\gcd(U_j,l)$ and thus each $i \in  L_{k_1}$ is in a cycle of size $\lambda_j l/d_j$ where $d_j=\gcd(U_j,l)$. Furthermore, the elements of $L_{k_1} \cup L_{k_2}\cup\cdots \cup L_{k_{\lambda_j}}$ are partitioned into $d_j$ cycles each of size $\lambda_j l/d_j$. Finally, since each $V_k$ is uniformly distributed in $\integers_l$, $U_j$ is also uniformly distributed.  Thus, as stated in Proposition~\ref{prop:lumped-parts}, every cycle length $\lambda_j$ contributes $d_j$ cycles of size $\lambda_j l/d_j$ to the cycle type of $\tau$.
\end{proof}
\begin{algorithm*}[t]
    \caption{Transpose of a partition}
    \begin{algorithmic}[1]
    \Require{$(a_l)_{l=1}^n$ (partition of $n$)}
    \Ensure{$(b_l)_{l=1}^n$ (transpose of $(a_l)_{l=1}^n$)}
    \State{$b_l \gets 0$ for all $l \in [n]$}
    \State{$L \gets \{l : a_l \neq 0\}$}
    \State{Sort $L\gets (l_1, l_2 ,\ldots , l_m)$ in descending order}
    \State{$k \gets 0$}
    \For{$i \gets 1$ to $m-1$}
        \State{$k \gets k + a_{l_i}$}
        \State{$b_k \gets l_i - l_{i+1}$}
    \EndFor
    \State{$k \gets k + a_{l_m}$}
    \State{$b_k \gets l_m$}
    \State \Return{$(b_l)_{l=1}^n$}
\end{algorithmic}
\label{alg:transpose}
\end{algorithm*}

\section{Algorithm for computing the transpose of a partition}

The following algorithm computes the transpose of a permutation in exponential notation. Let $a=(a_l)_{l=1}^n$ be the input initial partition and let $(b)=(b_l)_{l=1}^n$ be the transpose of $a$. Suppose that we sort the indices $l$ of $a$ such that $a_l\neq 0$. This gives the list $L=(l_1,l_2,\ldots,l_m)$ with $l_1 > l_2 >\cdots > l_m$ where $m$ is the number of non-zero entries of $a_l$. Also set $l_{m+1}=0$. The non-zero entries of $b$ are then 
\[
    b_{k_i} = l_i - l_{i+1},  
\]
where $k_i = a_{l_1}+a_{l_2}+\cdots+a_{l_i}$ for $i$ between $1$ and $m$. For example, if $a_{l_1}=1$, then the number of $1$'s in $b$ is equal to the difference in size between the two largest parts of $a$. These non-zero values of $b$ can be computed efficiently by going through the list $L$ and updating the partial sums $a_{l_1}+a_{l_2}\cdots + a_{l_i}$. Algorithm~\ref{alg:transpose} implements this idea.

Recall that the number of non-zero elements of $a$ is always at most $2\sqrt{n}$. Thus, the collection $L$ in Algorithm~\ref{alg:partitions} has size at most $2\sqrt{n}$ and the time complexity of sorting $L$ is $O(\sqrt{n}\log n)$. Updating $k$ and computing the difference $l_i-l_{i+1}$ has complexity $O(\log n)$. Thus, the complexity of Algorithm~\ref{alg:transpose} is $O(\sqrt{n}\log n)$. Which is comparable to the $O(\sqrt{n}(\log n)^2)$ complexity of the lumped Burnside process for partitions.




\end{appendices}





\end{document}